\documentclass{article}
\usepackage[utf8]{inputenc}

\usepackage[margin=1in]{geometry}
\usepackage{amsmath}
\usepackage{amsfonts}
\usepackage{amssymb}
\usepackage{amsthm}
\usepackage{comment}
\usepackage{algorithm}
\usepackage{algorithmicx}
\usepackage{algpseudocode}
\usepackage{enumerate}
\usepackage{mathtools}
\usepackage{mathrsfs}
\usepackage{hyperref}
\usepackage{graphicx}
\usepackage{verbatim}
\usepackage{float}
\usepackage{chngcntr}
\usepackage{xr}
\usepackage[shortlabels]{enumitem}

\usepackage{tikz}
\usepackage{cite}
\usepackage{hyperref}
\usepackage[capitalize, nameinlink]{cleveref}

\newcommand{\declarecolor}[2]{\definecolor{#1}{RGB}{#2}\expandafter\newcommand\csname #1\endcsname[1]{\textcolor{#1}{##1}}}
\declarecolor{White}{255, 255, 255}
\declarecolor{Black}{0, 0, 0}
\declarecolor{LightGray}{216, 216, 216}
\declarecolor{Gray}{127, 127, 127}
\declarecolor{Orange}{237, 125, 49}
\declarecolor{LightOrange}{251,229, 214}
\declarecolor{Yellow}{255, 192, 0}
\declarecolor{LightYellow}{255, 242, 200}
\declarecolor{Blue}{91, 155, 213}
\declarecolor{LightBlue}{222, 235, 247}
\declarecolor{Green}{112, 173, 71}
\declarecolor{LightGreen}{226, 240, 217}
\declarecolor{Navy}{68, 114, 196}
\declarecolor{LightNavy}{218, 227, 243}

\hypersetup{
	colorlinks=true,
	pdfpagemode=UseNone,
	citecolor=Green,
	linkcolor=Navy,
	urlcolor=Black,
	pdfstartview=FitW
}

\algnewcommand\algorithmicforeach{\textbf{for each}}
\algdef{S}[FOR]{ForEach}[1]{\algorithmicforeach\ #1\ \algorithmicdo}

\def\R{{\mathbb{R}}}

\theoremstyle{plain}
\newtheorem{theorem}{Theorem}[section]

\newtheorem{corollary}[theorem]{Corollary}
\newtheorem{lemma}[theorem]{Lemma}

\newtheorem{conjecture}[theorem]{Conjecture}

\newtheorem{claim}[theorem]{Claim}

\theoremstyle{remark}
\newtheorem{rem}[theorem]{Remark}

\theoremstyle{definition}
\newtheorem{defn}[theorem]{Definition}

\counterwithin{figure}{section}

\DeclareMathOperator{\poly}{poly} 

\newcommand\abs[1]{\left\lvert#1\right\rvert}

\newcommand\norm[1]{\left\lVert#1\right\rVert}

\newcommand\eps{\varepsilon}
\newcommand{\grad}{\nabla}

\begin{document}
	\title{Computing Circle Packing Representations of Planar Graphs}
	\author{
	Sally Dong
	\\University of Washington
	\and
	Yin Tat Lee\\
	University of Washington \\
	and Microsoft Research Redmond
	\and
	Kent Quanrud\\
	University of Illinois Urbana-Champaign
	}
	\maketitle
	\begin{abstract}
		The Circle Packing Theorem states that every planar graph can be represented as the tangency graph of a family of internally-disjoint circles. A well-known generalization is the Primal-Dual Circle Packing Theorem for 3-connected planar graphs. The existence of these representations has widespread applications in theoretical computer science and discrete mathematics; however, the algorithmic aspect has received relatively little attention. In this work, we present an algorithm based on convex optimization for computing a primal-dual circle packing representation of maximal planar graphs, i.e.\ triangulations. This in turn gives an algorithm for computing a circle packing representation of any planar graph. Both take $\widetilde{O}(n \log(R/\varepsilon))$ expected run-time to produce a solution that is $\varepsilon$ close to a true representation, where $R$ is the ratio between the maximum and minimum circle radius in the true representation.
	\end{abstract}

\maketitle

\section{Introduction}
	Given a planar graph $G$, a \emph{circle packing representation} of $G$ consists of a set of radii $\{r_v : v \in V(G) \}$ and a straight line embedding of $G$ in the plane, such that 
	\begin{enumerate}
		\item For each vertex $v$, a circle $C_v$ of radius $r_v$ can be drawn in the plane centered at $v$,
		\item all circles' interiors are disjoint, and
		\item two circles $C_u, C_v$ are tangent if and only if $uv \in E(G)$.
	\end{enumerate}

	\begin{figure}[H]
		\centering 
		\includegraphics[scale=0.4]{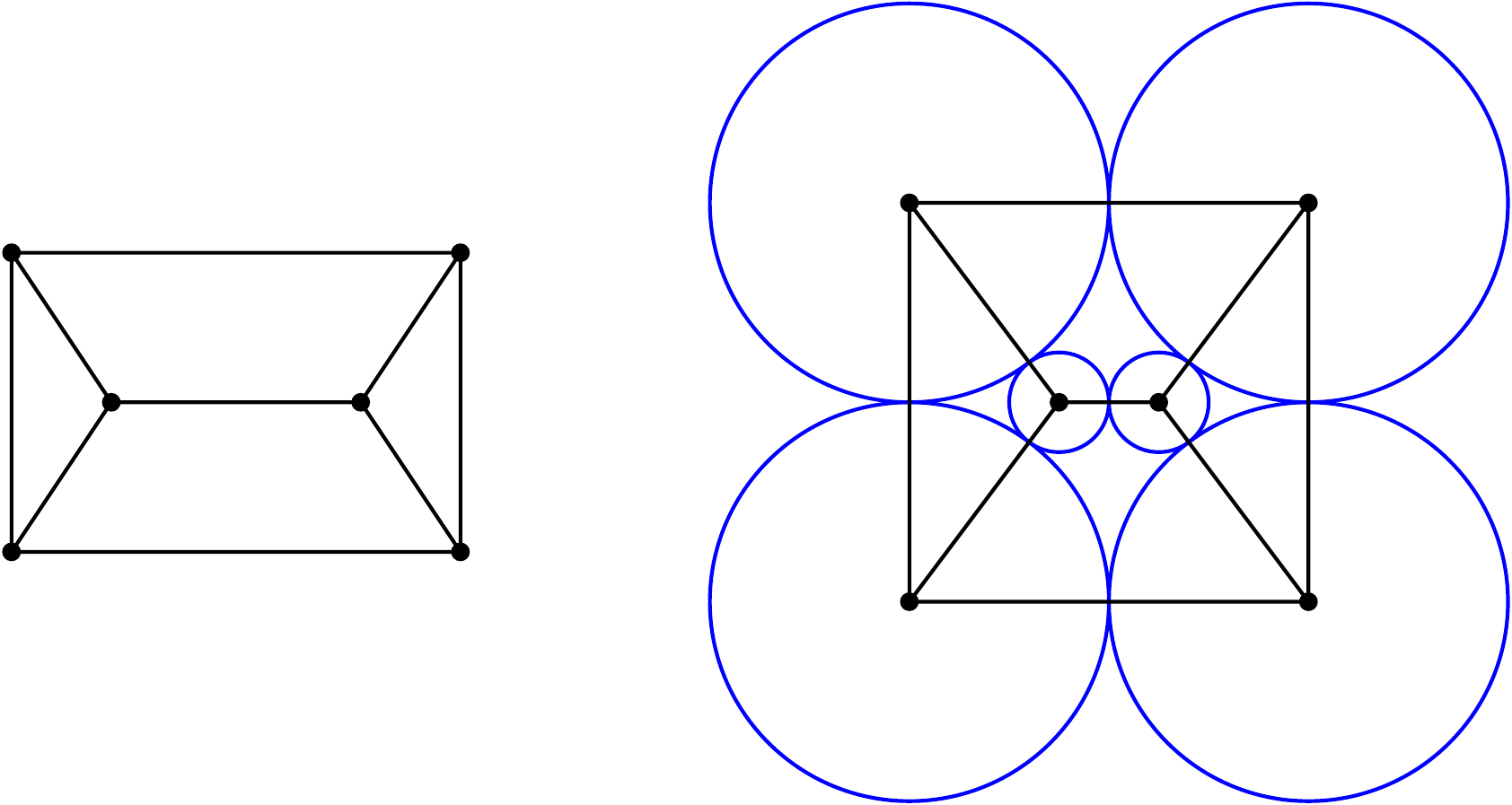}
		\caption{Example of a planar graph $G$ and its circle packing representation.}
	\end{figure}
	It is easy to see that any graph with a circle packing representation is planar. Amazingly, the following deep and fundamental theorem asserts the converse is also true.

	\begin{theorem}[Koebe-Andreev-Thurston Circle Packing Theorem~\cite{Koebe, Andreev, Thurston}] \label{thm:KAT}
		Every planar graph $G$ admits a circle packing representation. Furthermore, if $G$ is a triangulation then the representation is unique up to M\"obius transformations.
	\end{theorem}
	
	Recall that every embedded planar graph has an associated planar dual graph, where each face becomes a vertex and each vertex a face. In this paper, we will primarily focus on primal-dual circle packing, which intuitively consists of two circle packings, one
	for the original (primal) graph and another for the dual, that interact in a specific way. Formally:

	\begin{defn}[Simultaneous Primal-Dual Circle Packing] \label{def:pdcp}
		Let $G$ be a 3-connected planar graph, and $G^*$ its planar dual. Let $f_\infty$ denote the unbounded face in a fixed embedding of $G$; it also naturally identifies a vertex of $G^*$.
		
		The \emph{(simultaneous) primal-dual circle packing representation of $G$ with unbounded face $f_\infty$} is a set of numbers $\{r_v \; : \; v \in V(G)\} \cup \{r_f \; : \; f \in V(G^*)\}$ and straight-line embeddings of $G$ and $G^* - f_\infty$ in the plane such that:
		\begin{enumerate}
			\item $\{r_v \; : \; v \in V(G)\}$ is a circle-packing of $G$ with circles $\{C_v \; : \; v \in V(G)\}$,
			\item $\{r_f \; : \; f \in V(G^*) - f_\infty \}$ is a circle packing of $G^* - f_\infty$ with circles $\{C_f \; : \; f \in V(G^*) - f_\infty \}$
			\item $C_{f_\infty}$, the circle corresponding to $f_\infty$, has radius $r_{f_\infty}$ and contains $C_f$ for all $f \in V(G^*)$ in the plane. Furthermore, $C_{f_\infty}$ is tangent to $C_g$ if and only if $f_\infty g \in E(G^*)$.
			\item The two circle packing representations can be overlaid in the plane such that dual edges cross at a right angle, and no other edges cross. Furthermore, if $uv \in E(G)$ and $fg \in E(G^*)$ are a pair of dual edges, then $C_u$ is tangent to $C_v$ at the same point where $C_f$ is tangent to $C_g$.
		\end{enumerate}
	\end{defn}
	\begin{figure}[H]
		\centering 
		\includegraphics[scale=0.6]{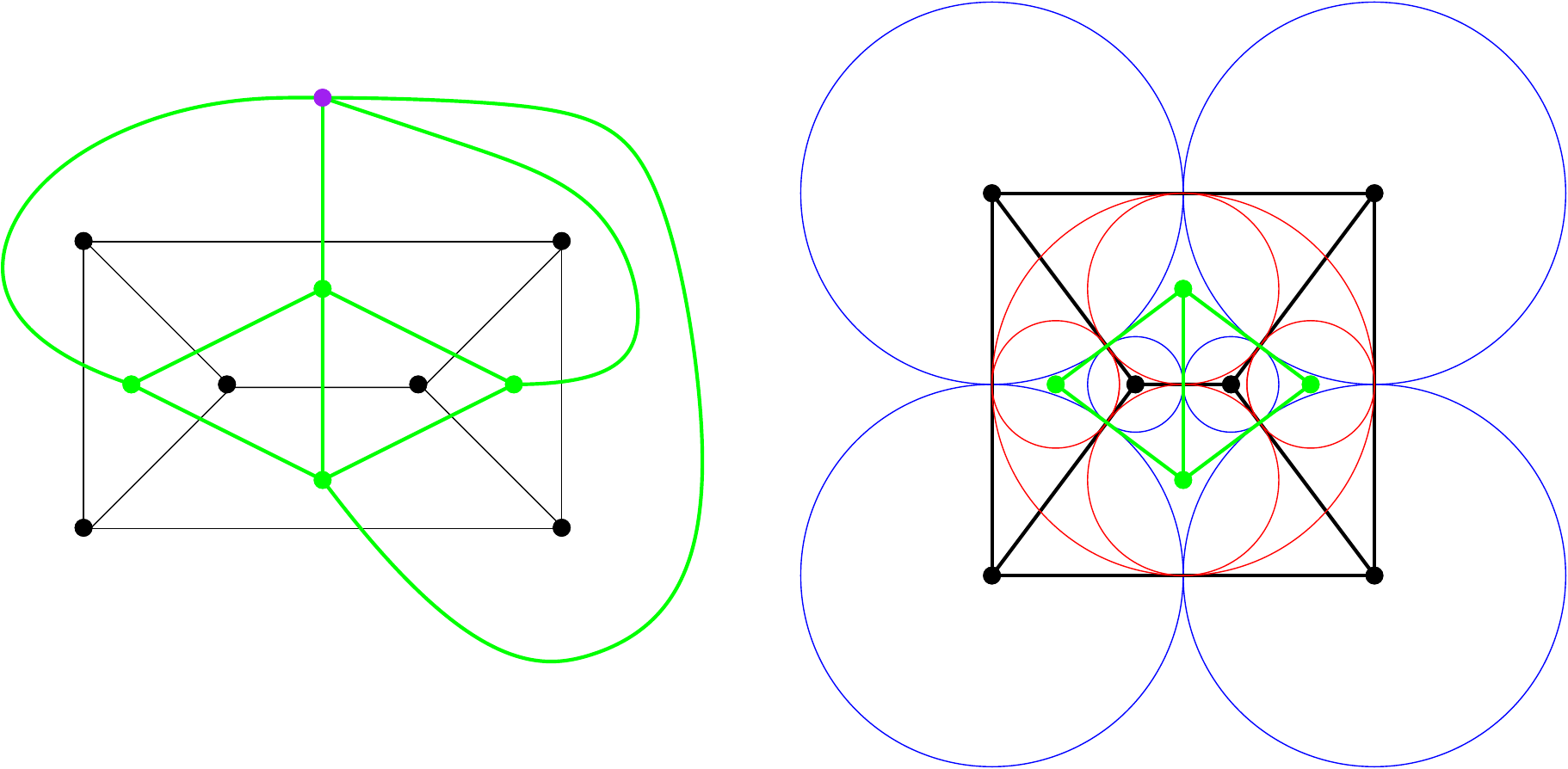}
		\caption{The planar graph $G$ from the previous example is in black on top. Its planar dual, $G^*$, is overlaid in green. The vertex corresponding to the unbounded face $f_\infty$ is marked in purple. On the bottom is the simultaneous primal-dual circle packing representation of $G$. The largest red circle is $C_{f_\infty}$.}
	\end{figure}
	
	Section~\ref{subsec:mobius-transform} provides more geometric intuitions regarding the definition. 
	
	The Circle Packing Theorem is generalized in this framework by Pulleyblank and Rote (unpublished), Brightwell and Scheinerman~\cite{Brightwell-Scheinerman}, and Mohar~\cite{Mohar93a}:
	\begin{theorem} \label{thm:SPDCP-existence}
		Every 3-connected planar graph admits a simultaneous primal-dual circle packing representation. Furthermore, the  representation is unique up to M\"obius transformations.
	\end{theorem}
	
	A primal-dual circle packing for a graph $G$ naturally produces a circle packing of $G$ by ignoring the dual circles.
	It also has a simple and elegant characterization based on angles in the planar embeddings, which we discuss in detail later. Moreover, the problem instance does not blow up in size compared to the original circle packing, since the number of faces is on the same order as the number of vertices in planar graphs.
	
	We remark here that either the radii vectors or the embeddings suffice in defining the primal-dual circle packing representation: Given the radii, the locations of the vertices are uniquely determined up to isometries of the plane; the procedure for computing them is discussed in Section~\ref{subsec:angle-graph}. Given the embedding, the radii are determined by the tangency requirements in Condition (4) of Definition~\ref{def:pdcp}.
	
	\subsection{Related Works and Applications} 
	Circle packing representations have many connections to theoretical computer science and mathematics. The Circle Packing Theorem is used in the study of vertex separators: It gives a geometric proof of the Planar Separator Theorem of Lipton and Tarjan~\cite{Miller-et-al, Har-Peled}; an analysis of circle packing properties further gives an improved constant bound for the separator size~\cite{spielman1996disk}; it is also used crucially to design a simple spectral algorithm for computing optimal separators in graphs of bounded genus and degree~\cite{kelner2006spectral}.
	In graph drawings, these representations give rise to straight-line planar embeddings; the existence of simultaneous straight-line planar embeddings of the graph and its dual, in which dual edges are orthogonal, was first conjectured by Tutte in his seminal paper in the area~\cite{Tutte}. 
	They are also used to prove the existence of Lombardi drawings and strongly monotone drawings for certain classes of graphs~\cite{eppstein2014mobius, felsner2016strongly}.  Benjamini used the Circle Packing Theorem as a key component in his study of distributional limits of sequences of planar graphs~\cite{Benjamini}. In polyhedral combinatorics, Steinitz's Theorem states that a graph is formed by the edges and vertices of a 3-dimensional convex polyhedron if and only if it is a 3-connected planar graph. The theorem and its generalization, the Cage Theorem, can be proved using the (Primal-Dual) Circle Packing Theorem~\cite{Ziegler}. For a more comprehensive overview of the other related works, see Felsner and Rote~\cite{Felsner-Rote}.
	 
	In Riemannian geometry, circle packing of triangulations is tightly connected to the Riemann Mapping Theorem, which states that there is a conformal (angle-preserving) mapping between two simply connected open sets in the plane. Thurston had conjectured that circle packings can be used to construct approximate conformal maps; this was later proved by Rodin and Sullivan~\cite{RodinSullivan}, which formed the basis of extensive work in discrete conformal mappings~\cite{he1996convergence} and analytic functions~\cite{Benjamini1996, dubejko1995circle,stephenson2002circle}. An excellent high-level exposition of this research direction is given by Stephenson~\cite{Stephenson}. One unique and important application is in neuroscience research: Conformal maps, and specifically their approximations using circle packings, can be used to generate brain mappings while preserving structural information~\cite{gu2008computational,hurdal2009discrete}. This suggests a real-world interest in efficient circle packing algorithms.
	
	Computationally, Bannister et al.\ ~\cite{bannister2014galois} showed that numerical approximations of circle packing representations are necessary. Specifically, they proved for all large $n$, there exists graphs on $n$ vertices whose exact circle packing representations involve roots of polynomials of degree $\Omega(n^{0.677})$; solving these exactly, even under extended arithmetic models, is impossible. Mohar~\cite{Mohar93a, MoharEuclidean} gave a polynomial-time iterative algorithm to compute $\eps$-approximations of primal-dual circle packings in two phases: the radii are approximated first, followed by the position of the vertices. The presentation was very recently simplified by Felsner and Rote~\cite{Felsner-Rote}. However, because run-time was not the focus beyond demonstrating that it is polynomial, a rudimentary analysis of the algorithm puts the complexity at $\widetilde{\Omega}(n^5)$.
	For general circle packing, Alam et al.\ ~\cite{alam2014balanced} gave algorithms with a more combinatorial flavour for special classes of graphs, including trees and outerpaths in linear time, and fan-free graphs in quadratic time. Chow~\cite{chow2003combinatorial} showed an algorithm based on Ricci flows that converges exponentially fast to the circle packing of the triangulation of a closed surface.
	
	In practice, for general circle packing, there is a numerical algorithm $\texttt{CirclePack}$ by Stephenson which takes a similar approach as Mohar and works well for small instances~\cite{Collins-Stephenson-Alg}. The current state-of-the-art is by Orick, Collins and Stephenson~\cite{LinearizedAlg}; here the approach is to alternate between adjusting the radii and the position of the vertices at every step. The algorithm is implemented in the $\texttt{GOPack}$ package in MATLAB; numerical experiments using randomly generated graphs of up to a million vertices show that it performs in approximately linear time. However, there is no known proof of convergence.
	
	\subsection{Our Contribution}
	We follow the recent trend of attacking major combinatorial problems using tools from convex optimization. Although the combinatorial constraints on the radii had been formulated as a minimization problem in the past (e.g.\ by Colin de Vedi\`ere~\cite{de1991principe} and Ziegler~\cite{Ziegler}), the objective function is ill-conditioned, and therefore standard optimization techniques would only give a large polynomial time. Our key observation is that the primal-dual circle packing problem looks very similar to the minimum $s$-$t$ cut problem when written as a function of the logarithm of the radii (see Equation~\eqref{eqn:convex-objective}). Due to this formulation, we can combine recent techniques in interior point methods and Laplacian system solvers~\cite{koutis2011nearly, kelner2013simple, lee2013efficient, spielman2014nearly, koutis2014approaching, cohen2014solving, peng2014efficient, kyng2016sparsified, kyng2016approximate} to get a run-time of $\widetilde{O}(n^{1.5} \log R)$, where $R$ is the ratio between the maximum and minimum radius of the circles. In the worst case, this ratio can be exponential in $n$; however the approach still gives a significantly improved run-time of $\widetilde{O}(n^{2.5})$.
	
	For further improvements, our starting point is the recent breakthrough on matrix scaling problems \cite{cohen2017matrix}, which showed that certain class of convex problems can be solved efficiently using vertex sparsifier chains \cite{kyng2016sparsified}. When applied to the primal-dual circle packing problem, it gives a run-time of $\widetilde{O}(n \log^2 R)$, which is worse than interior point in the worst case. One of the $\log R$ term comes from the accuracy requirement for circle packing; this term seems to be unavoidable for almost all existing iterative techniques. The second term comes from the problem diameter.
	
	To obtain a better bound, we present new properties of the primal-dual circle packing representation for triangulations using graph theoretic arguments. In particular, we show there is a spanning tree in a related graph such that the radii of  neighbouring vertices are polynomially close to each other (Section~\ref{subsec:good-tree}). This allows us to show that the objective function is locally strongly convex (Lemma \ref{lem:strongly_convex}). Combining this with techniques in matrix scaling \cite{cohen2017matrix}, we achieve a run-time of $\widetilde{O}(n \log R)$ for primal-dual circle packing for triangulations and general circle packing. Given the problem requires minimizing a convex function with accuracy $1/R$, we attain the natural run-time barrier of existing convex optimization techniques.
	
	\subsection{Our Result}
	For primal-dual circle packing, we focus on triangulations, which are maximal planar graphs, and present a worst-case nearly quadratic time algorithm.
	
	\begin{theorem}\label{thm:main}
		Let $G$ be a triangulation where $|V(G)| + |V(G^*)| = n$, and let $f_\infty \in V(G^*)$ denote its unbounded face. There is an explicit algorithm that finds radii $r \in \R^n$ with $r_{f_\infty} = 1$, and locations $p \in \R^{2(n-1)}$ of $V(G) \cup V(G^*) - f_\infty$ in the plane, such that
		\begin{enumerate}
			\item there exists a target primal-dual circle packing representation of $G$ with radii vector $r^* \in \R^n$ and vertex locations $p^* \in \R^{2(n-1)}$; furthermore, $r^*_{f_\infty} = 1$ and $\norm{r^*}_\infty = O(1)$,
			\item $1 - \eps \leq r_u/r^*_u \leq 1 + \eps$ for each $u \in V(G) \cup V(G^*)$, and
			\item $\norm{p_u - p_u^*}_\infty \leq \eps/R$ for each $u \in V(G) \cup V(G^*) - f_\infty,$
		\end{enumerate}
		where $R = r^*_{\max}/r^*_{\min}$ is the ratio between the maximum and minimum radius in the target representation.
		The algorithm is randomized and runs in expected time
		\[
		\widetilde{O}\left(n \log \frac{R}{\eps}\right).
		\]
	\end{theorem} 
	
	\begin{rem}
		We use the more natural $\eps/R$ for the location error instead of $\eps$, in order to reflect the necessary accuracy at the smallest circle, which has radius $\Theta(1/R)$. We use $\widetilde{O}$ in the runtime to hide a $\poly(\log(n), \log \log(R/\eps))$ factor. 
	\end{rem}
	
	From Theorem~\ref{thm:main}, an algorithm for general circle packing is easily obtained.
	 
	\begin{theorem}\label{cor:circle-packing}
		Let $G$ be \emph{any} planar graph where $|V(G)| = n$. There is an explicit algorithm that finds radii $r \in \R^n$ and locations $p \in \R^{2n}$, such that
		\begin{enumerate}
			\item there exists a target circle packing of $G$ with radii vector $r^* \in \R^n$ and vertex locations $p^* \in \R^{2n}$, and $\norm{r}_\infty = O(1)$,
			\item $1 - \eps \leq r_u/r^*_u \leq 1 + \eps$ for each $u \in V(G)$, and
			\item $\norm{p_u - p_u^*}_\infty \leq \eps/R$ for each $u \in V(G),$
		\end{enumerate}
	where $R = r^*_{\max}/r^*_{\min}$ is the ratio between the maximum and minimum radius in the target representation.
	The algorithm is randomized and runs in expected time
	\[
	\widetilde{O}\left(n \log \frac{R}{\eps}\right).
	\]
	\end{theorem}
	
	\begin{rem}$R$ is a natural parameter of the circle packing problem and is $\poly(n)$ for several classes of graphs as described in~\cite{alam2014balanced}; it is bounded by $(2n)^{n}$ in the worst case (Corollary~\ref{cor:min-max-radii-ratio}). When $R$ is $\poly(n)$, our algorithm achieve nearly linear-time complexity.\end{rem}

\section{Solution Characterization}\label{sec:graph}

	In this section, we present some structural properties of primal-dual circle packing representations which will be crucial to the algorithm. We begin with a review of basic graph theory concepts.
	
	A \emph{plane graph} is a planar graph $G$ with an associated planar embedding. The embedding encodes additional information beyond the vertex and edge sets of $G$; in particular, it defines the faces of $G$ and therefore a \emph{cyclic ordering} of edges around each vertex. It is folklore that any 3-connected planar graph has a well-defined set of faces.
	
	The \emph{dual graph} of a plane graph $G$ is denoted by $G^*$. Its vertex set is the set of faces of $G$, and two vertices are adjacent in $G^*$ whenever the corresponding faces in $G$ share a common edge on their boundary. Note that there is a natural bijection between the edges of $G^*$ and the edges of $G$. We denote the unbounded face of a plane graph $G$ by $f_\infty$.
	
	\begin{figure}[H]
		\centering 
		\includegraphics[scale=0.6]{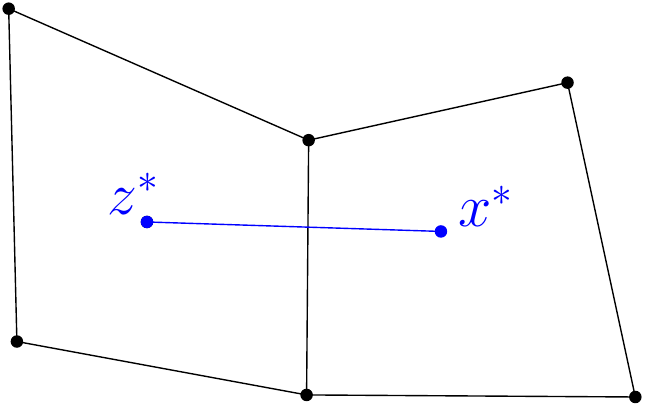}
		\caption{Two bounded faces $z,x$ of a plane graph $P$ is shown in black; they correspond to dual vertices $z^*$ and $X^*$, in blue. The dual of the edge $z^*x^* \in E(P^*)$ is the unique edge that $z^*x^*$ crosses in the embedding; note that it is on the boundaries of both faces $z$ and $x$ in $P$.}\label{fig:duality}
	\end{figure}

	\subsection{Representations on the Extended Plane}\label{subsec:mobius-transform}
	In the definition of primal-dual circle packing representation, we specified that the embeddings are in the Euclidean plane. For a more intuitive view, consider the embeddings in the extended plane with the appropriate geometry: Here, Conditions (2) and (3) in Definition~\ref{def:pdcp} collapse into one, which asks for a valid circle packing of $G^*$ in the extended plane, such that $C_{f_\infty}$ is a circle centered at infinity (its radius becomes irrelevant). All other tangency requirements hold as before, and the interaction between the primal and dual embeddings are not changed.
	
	This view ties into M\"obius transforms, mentioned in Theorems~\ref{thm:KAT} and~\ref{thm:SPDCP-existence}. A \emph{M\"obius transform} is an angle-preserving map of the extended plane to itself; moreover, it maps circles to lines or circles. It can be shown that for any two faces $f,g$ of $G$, a primal-dual circle packing representation of $G$ with unbounded face $f$ can be obtained from one with unbounded face $g$ via an appropriately defined M\"obius transform. Furthermore, the roles of $G$ and $G^*$ become interchangeable.
	
	For our algorithms, we compute a primal-dual circle packing representation of $G$ after fixing an unbounded face, and do not concern ourselves with these transforms. We continue with the original definition of embedding in the Euclidean plane.
	
	\subsection{Angle Graph}\label{subsec:angle-graph}
	Given a 3-connected plane graph $G$, the \emph{angle graph} of $G$ is the bipartite plane graph $\hat{H}_G = (V(G) \cup V(G^*), E(\hat{H}))$ constructed as follows: For each vertex $v \in V(G)$, fix its position in the plane based on $G$; place a vertex $f$ in each face of $G$ (including the unbounded face $f_\infty$); connect $v, f \in V(\hat{H})$ with a straight line segment if and only if $v$ is a vertex on the boundary of $f$ in $G$. When the original graph $G$ is clear, we simply write $\hat{H}$. It is convenient to also define the \emph{reduced angle graph} $H$, obtained from $\hat{H}$ by removing the vertex corresponding to $f_\infty$. $H$ is again a bipartite plane graph; all its bounded faces are of size four.

	The (reduced) angle graph is so named because of the properties that become apparent when its embedding derives from a primal-dual circle packing representation of $G$: Specifically, suppose $r,p$ are the radii and location vectors of a valid representation, and that the locations of vertices of $H$ are given by $p$. Note that $G$'s outer cycle $C_o  = (s_1, \dots, s_k)$ must be embedded as a convex polygon, in order for conditions on $C_{f_\infty}$ to be satisfied; suppose the polygon has interior angle $\alpha_i$ at vertex $s_i$. Then for any $u \in V(H)$,
	\begin{equation}\label{eqn:angle-constraints}
		\sum_{w \; : \; uw \in E(H)} \arctan \frac{r_w}{r_u} = \begin{cases}
		\pi \qquad &u \notin C_o\\
		\alpha_i/2 \qquad &u \in C_o.
		\end{cases}
	\end{equation}
	
	To see this, first observe that an edge $uw$ in this embedding has a natural \emph{kite} $K_{uw}$ in the plane associated with it, formed by the vertices $u, w$ and the two intersection points of $C_u$ and $C_w$. (See, for example, edge $vf$ in Figure~\ref{fig:kite-in-H}.) Furthermore, distinct kites do not intersect in the interior. Suppose $u \notin C_o$, and let $w_1, \dots, w_l$ denote its neighbours in cyclic order. Then $C_v$ is covered by the kites $K_{uw_1}, \dots, K_{uw_l}$, which all meet at the vertex $u$ and are consecutively tangent. Each neighbour $w_i$ contributes an angle of $2 \arctan(r_{w_i}/r_u)$ at $u$ for a total of $2\pi$. For the vertices on $C_o$, it can be shown that if $u = s_i$, the kites will cover an angle equal to $\alpha_i$.
	
	\begin{figure}[H]
		\centering 
		\includegraphics[scale=0.6]{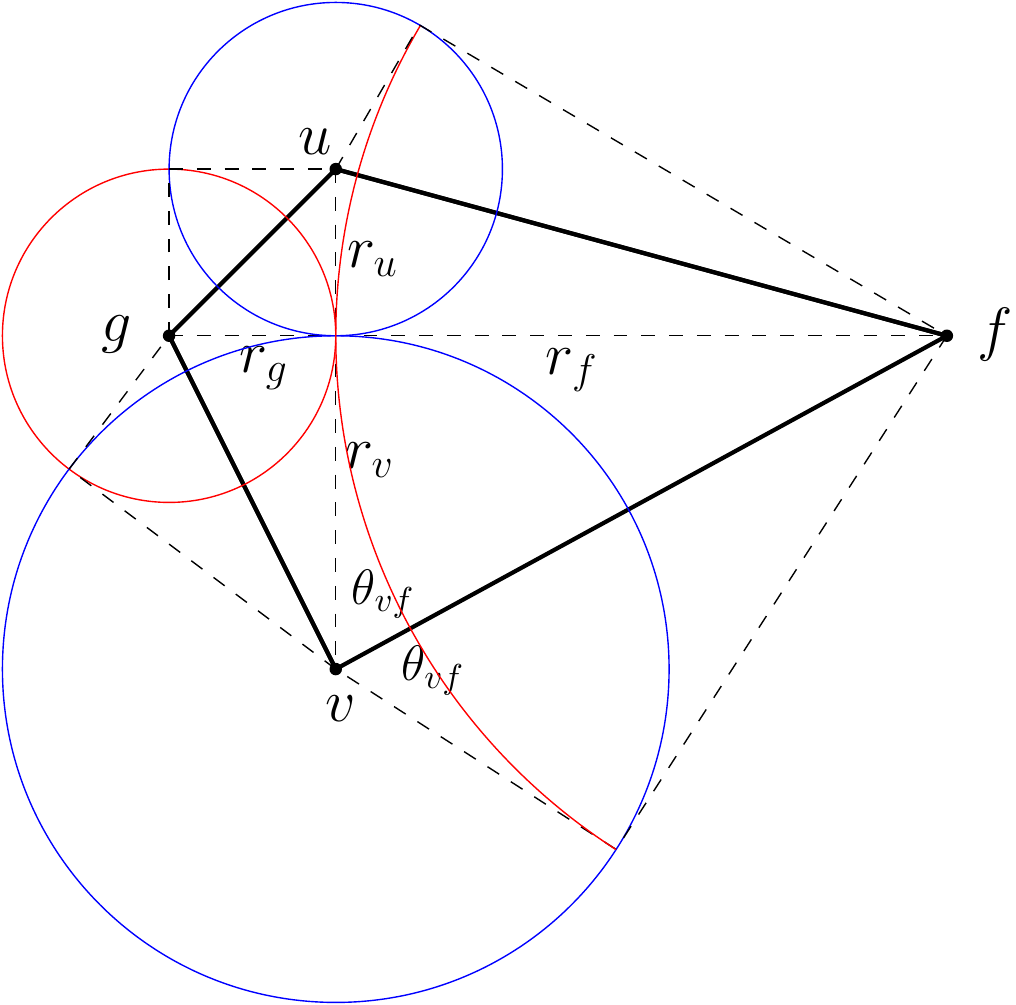}
		\caption{An illustration of the structure of $H$ locally, with edges of $H$ shown in black. Primal circles are in blue and dual circles in red. Vertices $u,f,v,g$ define the boundary of a face; $u,v \in V(G)$ and $f,g \in V(G^*)$. The edges $uv \in E(G)$ and $fg \in E(G^*)$ are dual to each other and cross at a right angle, as required by the circle packing representation. Observe, for example, $C_v$ is partially covered by $K_{vg}$ and $K_{vf}$.}
		\label{fig:kite-in-H}
	\end{figure}
	
	Conversely, any $r \in \R^{|V(H)|}$ with the above property \emph{almost} suffices as the radii of a primal-dual circle packing representation. Indeed, we can embed $H$ (and therefore $G$ and $G^* - f_\infty$) based on $r$ as follows: Fix any vertex $u$ to start; embed the vertices in $N(u)$ in cyclic order around $u$, by forming the consecutively tangent kites using $r$. The process continues in a breadth-first fashion until all the vertices are placed. By construction, this embedding with radii $r$ satisfies conditions (1),(2),(4) in Definition~\ref{def:pdcp}. Moreover, the outer cycle of $G$ forms a convex $k$-gon with interior angles $\alpha_1, \dots, \alpha_k$. 
	
	The following theorem states that vectors $r$ satisfying Equation~\eqref{eqn:angle-constraints} must exist.
	\begin{theorem}[\cite{MoharEuclidean}]\label{thm:r-existence}
		Let $G$ be a 3-connected plane graph with outer cycle $C_o = (s_1, \dots, s_k)$ and unbounded face $f_\infty$. Let $H$ be its reduced angle graph, and $\alpha_1, \dots, \alpha_k \in (0, \pi)$ such that $\sum_i \alpha_i = (k-2)\pi$. Then, up to scaling, there exists a \emph{unique} $r \in \R^{|V(H)|}$ satisfying Equation~\eqref{eqn:angle-constraints}. 
	\end{theorem}

	For our purposes, $G$ is a triangulation with outer cycle $C_o = (s_1, s_2, s_3)$ and unbounded face $f_\infty$. By Theorem~\ref{thm:r-existence}, there exists $r \in \R^{|V(H)|}$ such that Equation~\eqref{eqn:angle-constraints} is satisfied with $\alpha_i = \pi/3$ for $i = 1,2,3$. This gives rise to a primal-dual circle packing representation without $C_{f_\infty}$, where the outer cycle $C_o$ is embedded as a triangle with interior angles all equal to $\pi/3$, i.e\ an equilateral triangle. It follows that all the $r_{s_i}$'s must be equal, and therefore we can take $C_{f_\infty}$ to be the unique circle inscribed in the outer  triangle, leading to an overall valid representation.
	
	This construction motivates the next definition. 
	\begin{defn}
		For a triangulation $G$ with outer cycle $C_o = (s_1, s_2, s_3)$ and unbounded face $f_\infty$, the \emph{$C_o$-regular primal-dual circle packing representation} of $G$ is the unique representation where $C_o$ is embedded as an equilateral triangle, and $C_{f_\infty}$ is the circle of radius 1 inscribed in the triangle.
	\end{defn}

	Our algorithm will therefore focus on finding the $C_o$-regular representation, using the characterization of the radii from Theorem~\ref{thm:r-existence}.


	\subsection{Existence of a Good Spanning Tree}\label{subsec:good-tree}
	Throughout this section, $G$ denotes a triangulation with outer cycle $C_o = (s_1, s_2, s_3)$ and unbounded face $f_\infty$; $r$  denotes the radii vector of the unique $C_o$-regular primal-dual circle packing of $G$; $\hat{H}$ denotes the angle graph of $G$; and $H$ the reduced angle graph.
	
	The $C_o$-regular circle packing representation of $G$ naturally gives rise to a simultaneous planar embedding of $G, G^*$, and $H$. All subsequent arguments will be in the context of this embedding.
	
	\begin{defn}
		A \emph{good edge in $\hat{H}$ with respect to $r$} is an edge $uw \in E(\hat{H})$ so that $1/(2n) \leq r_u/r_w \leq 2n$. A set of edges is good if each edge in the set is good. Predictably, what is not good is \emph{bad}.
	\end{defn}
	Since we examine the radius of a vertex in relation to those of its neighbours, the next definition is natural:
	\begin{defn}
		Let $u \in V(\hat{H})$. For any good edge $uw$, we say $w$ is a \emph{good} neighbour of $u$. For a bad edge $uw$, we say $w$ is a \emph{bad} neighbour of $u$; we further specify that $w$ is a \emph{large} neighbour if $r_w/r_u > 2n$ or a \emph{small} neighbour if $r_u/r_w > 2n$.
	\end{defn}

	Recall that $H$ is a bipartite graph, with vertex partitions $V(G)$ and $V(G^*) - f_\infty$. For the last piece of notation, we will call a vertex $u$ of $H$ a \emph{$V$-vertex} if it is in the first partition, and call $u$ an \emph{$F$-vertex} if it is in the second partition. 
	
	Our main theorem in this section is the following:
	
	\begin{theorem}\label{thm:good-tree}
		There exists a good spanning tree in $\hat{H}$ with respect to $r$.
	\end{theorem}
	
	\begin{proof}		
		First, we consider $f_\infty \in V(\hat{H})$: It has radius 1, and is inscribed in the equilateral triangle with vertices $C_o = \{s_1, s_2, s_3\}$. Hence, $r_{s_i} = \tan \frac{\pi}3$ for each $s_i \in C_o$. It follows that all of $f_\infty$'s incident edges in $\hat{H}$ are good, so any good spanning tree in $H$ extends to one in $\hat{H}$. 
		
		It remains to find a good spanning tree in $H$. To continue, we require the following lemma regarding a special circle packing structure. 
		
		\begin{lemma}\label{lem:circle-packing-structure}
			Let $C_1$ and $C_2$ be two circles with centers $X$ and $Y$ and radii $R_1, R_2$ respectively, and tangent at a point $P$. Suppose without loss of generality $R_2 \leq R_1$. Let $Q$ be a point of distance $R_2/n$ from $P$, so that $PQ$ and $XY$ are perpendicular. Let $L_1$ be a line segment parallel to $XY$ through $Q$ with endpoints on $C_1$ and $C_2$. Let $L_2$ be the line parallel to $XY$, further away from $XY$ than $L_1$, and tangent to $C_2$.
			
			Suppose we place a family $\mathcal{C} = \{D_1, \dots, D_m\}$ of $m$ internally-disjoint circles (of any radius) in the plane, where $m < n$, such that:
			\begin{enumerate}
				\item no circles from $\mathcal{C}$ intersect $C_1$ or $C_2$ in the interior,
				\item at least one circle from $\mathcal{C}$ intersects $L_1$, and
				\item the tangency graph of $\mathcal{C}$ is connected,
			\end{enumerate} 
			then all circles in $\mathcal{C}$ are contained in the region bounded by $C_1, C_2, L_2$.
		\end{lemma}
		
		\begin{figure}[H]
			\centering 
			\includegraphics[scale=0.4]{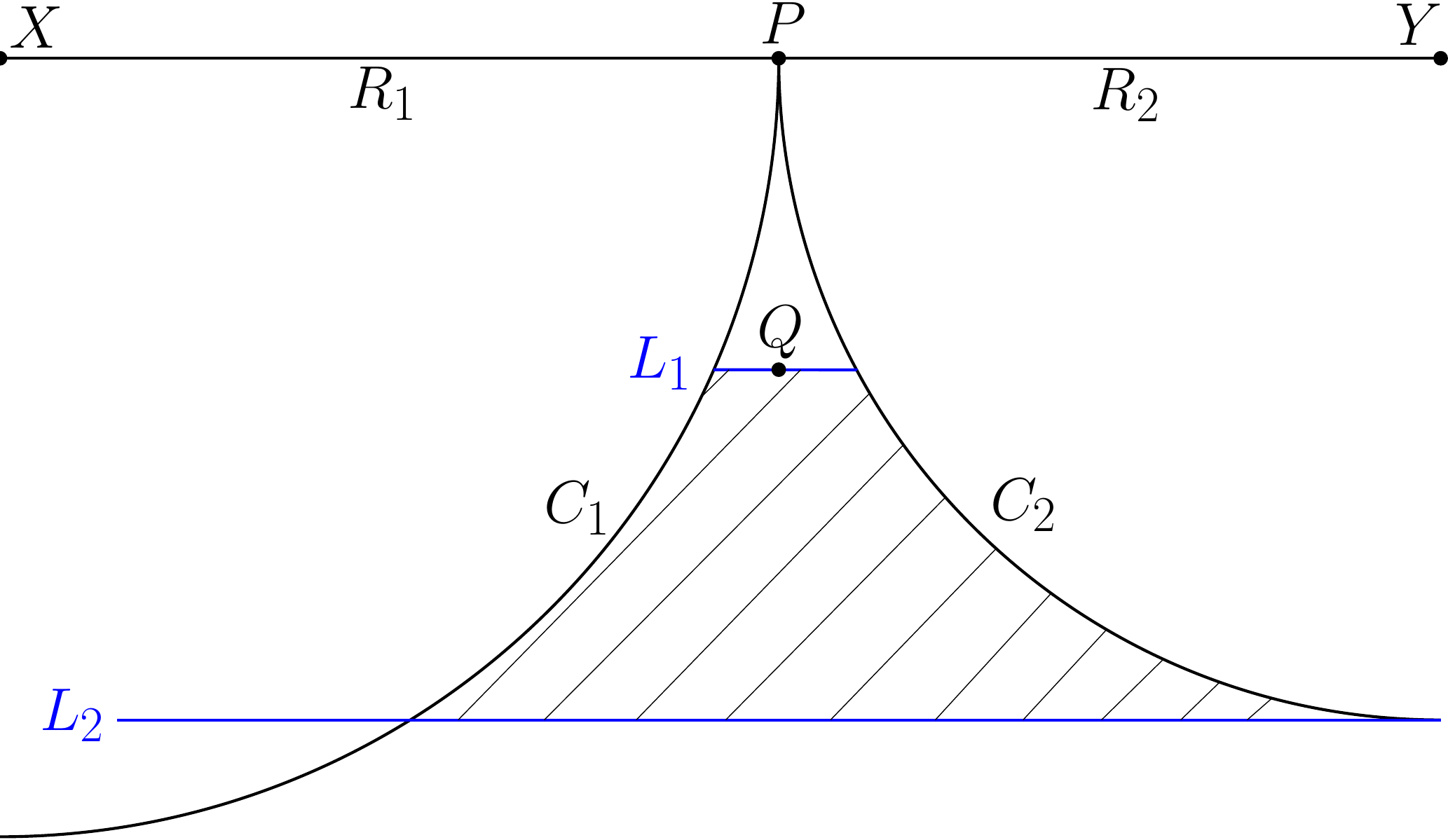}
			\caption{Illustration of Lemma~\ref{lem:circle-packing-structure}. Any family $\mathcal{C}$ of circles satisfying the conditions of the lemma must be contained in the shaded region. The diagram is slightly deceptive: in reality $Q$ is much closer to $P$ than depicted.}
			\label{fig:geometry}
		\end{figure}
		
		\begin{proof}
			Let $h(\mathcal{C})$ denote the maximum Euclidean distance between a point $x$ on a circle in $\mathcal{C}$ and the line $XY$; in other words, $h(\mathcal{C}) = \max \{d(x, XY) \; : \; x \in \bigcup_{D \in \mathcal{C}} C\}$. We want to show $h(\mathcal{C}) < R_2$. 
			
			Suppose we place the circles one at a time while maintaining tangency, starting with $D_1$ intersecting $L_1$. By elementary geometry, it is clear that each additional circle $D_i$ should be below $D_{i-1}$,  tangent to only $D_{i-1}$, and be of maximum size possible, i.e.\ tangent to both $C_1$ and $C_2$. Given this observation, it suffices to consider when $R_1 = R_2 = R$. 
			
			In this case, $h(\mathcal{C})$ is maximized when all the circles are arranged as described above, and have their centers on the line through $P$ and $Q$. Let $a_i$ be the radius of $D_i$ for each $i \in [m]$, and let $a_0 = d(P,Q)$. By the Pythagorean Theorem, we know that if $h_0 = d(P,Q)$ and $h_i = h(\{D_1, \dots, D_i\})$, then
			\[
			(a_{i+1} + h_i)^2 + R^2 = (a_{i+1} + R)^2.
			\]
			Hence we have the following recurrence relationship:
			\begin{align*}
			h_0 &= R/n \\
			h_{i+1} &= h_i + 2a_{i+1} = h_i\left(1+ \frac{h_i}{R - h_i}\right)
			\intertext{If we let $b_i = 2R/h_i$, then}
			b_0 &= 2n \\
			\frac{2R}{b_{i+1}} &= \frac{2R}{b_i}\left(1 + \frac{2}{b_i - 2}\right) \\
			b_{i+1} &= \frac{b_i}{1 + \frac{2}{b_i - 2}}\\
			&= b_i - 2.
			\end{align*}
			So $h_m = \frac{R}{b_m} = \frac{R}{2n - 2m} < R$, as desired.
		\end{proof}
		
		Recall that $r$ satisfy the following angle constraints for each $u \in V(H) - C_o$: 
		\begin{equation} \label{eqn:angle-constraint-again}
		\sum_{w \; : \; uw \in E(H)} \arctan \frac {r_w}{r_u} = \pi 
		\end{equation}
		
		\begin{claim}\label{claim:f-vertex-neighbour}
			Every $F$-vertex in $H$ can have at most 1 large neighbour. Furthermore, they have no small neighbours. Consequently, there are no $F$-vertices with only bad neighbours.
		\end{claim}
	
		\begin{proof}
			W require $G$ to be a triangulation, so that all $F$-vertices have degree three.
			
			Suppose $f$ is an $F$-vertex with 2 large neighbours $v_1, v_2$, and without loss of generality, $r_{v_1} \leq r_{v_2}$. By the angle constraints in Equation~\eqref{eqn:angle-constraint-again}, $f$'s third neighbour $u$ must be small.
			
			Consider the primal-dual circle packing locally around $f$: The circles $C_{v_1}, C_{v_2}$ are tangent at a point $P$, which is on the line $L$ connecting the centers of the two circles. Furthermore, $C_f$ is tangent to $L$ at the point $P$. By the definition of large neighbours, we know $r_f < r_1/2n$. Moreover, $C_u$ must intersect $C_f$, so $C_u$ is at a distance of at most $2r_f$ away from $P$. Now, let us restrict our attention to primal circles (which include $C_{v_1}, C_{v_2}, C_u$) and apply Lemma~\ref{lem:circle-packing-structure}.
			
			Let $N_G(v_1) = \{v_2, u = w_1, \dots, w_l\}$ denote the neighbours of $v_1$ in $G$ in cyclic order. Since $G$ is a triangulation, we know that $w_i w_{i+1} \in E(G)$ for each $i \in [l]$, and $w_lv_2 \in E(G)$. 
			This means in the primal circle packing, $C_{w_i}$ is tangent to $C_{w_{i+1}}$ for each $i$, and $C_{w_l}$ is tangent to $C_{v_2}$. There are two cases to consider: 
			\begin{enumerate}
				\item $v_1 \notin C_o$: 
				In this case, $v_1, v_2, w_l$ are the vertices of a bounded face of $G$. Hence, the circles $C_{v_2}, C_u, C_{w_1}, \dots, C_{w_l}$ are consecutively tangent and surround $C_{v_1}$. This contradicts the conclusion of Lemma~\ref{lem:circle-packing-structure}.
				
				\item $v_1 \in C_o$: 
				Note that the primal circles with the largest radii correspond to the vertices in $C_o$. Since $r_{v_2} \geq r_{v_1}$, we must have $r_{v_2} = r_{v_1}$ and $v_2 \in C_o$. Then $w_l \in C_o$ must be the third vertex on the boundary of $f_\infty$, with $v_1, v_2, w_l$ forming an equilateral triangle. Staring at the position of $C_{w_l}$, we see that this also contradicts the conclusion of Lemma~\ref{lem:circle-packing-structure}.
			\end{enumerate}
		
			\begin{figure}[H] 
				\centering 
				\includegraphics[scale=0.75]{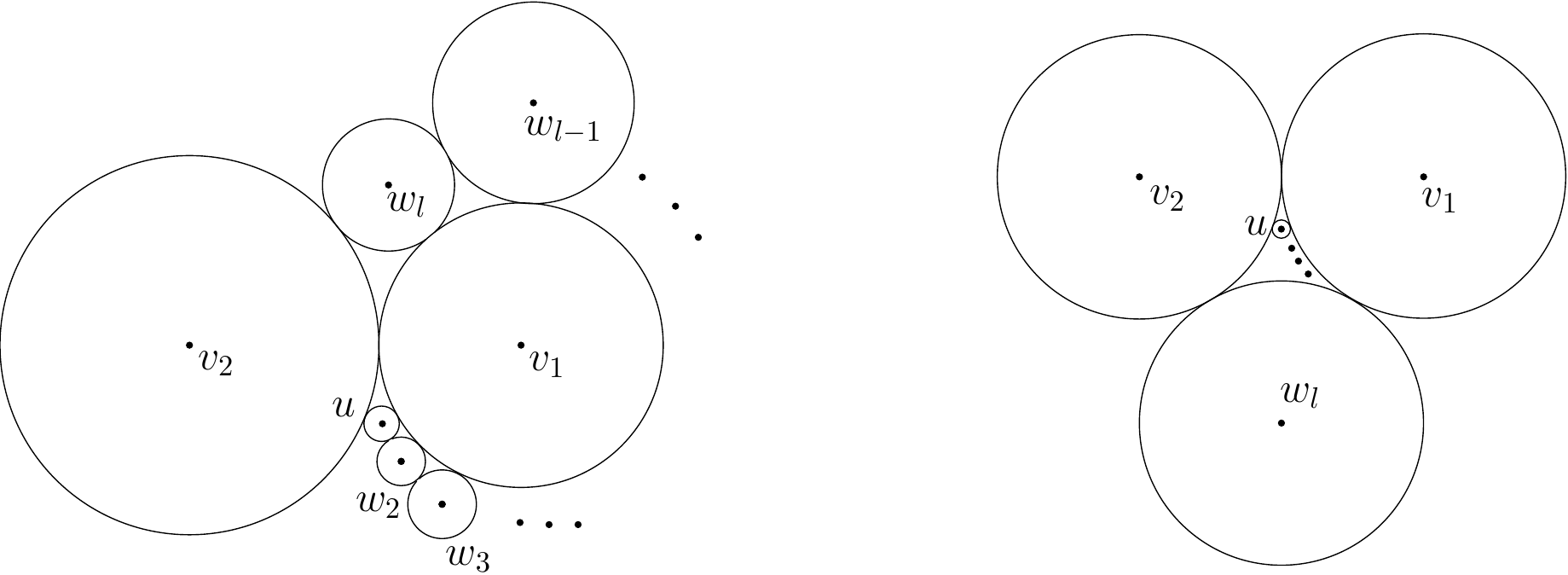}
				\caption{Illustration of the primal circles around $v_1$. There are two possible cases: in the former, circles corresponding to the neighbours of $v_1$ surround $C_{v_1}$; in the latter, they do not. Ellipses indicate additional $C_{w_i}$'s that are tangent to $C_{v_1}$.}
				\label{fig:f-neighbour-cases}
			\end{figure}
	
			So we have shown that $f$ has at most one large neighbour. 
			
			Suppose $f$ has a small neighbour $u$ and two other neighbours $v,w$, at most one of which is large. Then by the angle constraints in~\eqref{eqn:angle-constraint-again},
			\[
				\pi = \arctan \frac{r_u}{r_f} + \arctan \frac{r_v}{r_f} + \arctan \frac{r_w}{r_f} < \arctan n^{-2} + \arctan n^2 + \pi/2 = \pi,
			\]
			a contradiction.
		\end{proof}
	
		\begin{claim}\label{claim:no-isolated-v-vertex}
			There are no $V$-vertices in $H$ with only bad neighbours.
		\end{claim}
		\begin{proof}
			Suppose $v$ is a $V$-vertex with only bad neighbours. Again, by Equation~\eqref{eqn:angle-constraint-again}, two of its neighbours are large and the remaining are small. Let $f$ denote one of its large neighbours. But then $v$ is a small neighbour of $f$, contradicting the previous claim.	
		\end{proof}
		
		We have shown there are no vertices incident to only bad edges. Before proceeding, we observe the following:
		\begin{claim}\label{claim:H-boundary-good}
			Recall $C_o = (s_1, s_2, s_3)$ are vertices on the outer cycle of $G$. Let $t_1, t_2, t_3$ be $F$-vertices corresponding to faces in $G$ that are adjacent to $f_\infty$. 
			Let $B = (s_1, t_1, s_2, t_2, s_3, t_3) \subset V(H)$ be the set of vertices on the outer cycle of $H$. 
			Then all the edges of $H[B]$ are good.
		\end{claim}
		\begin{proof}
			Suppose without loss of generality that $s_i t_i$ is a bad edge. Since $s_1,s_2,s_3$ have equal radii, $t_i s_{i+1}$ must also be a bad edge. This contradicts Claim~\ref{claim:f-vertex-neighbour} which specifies that $t_i$ has no small neighbours and at most one large neighbour.
		\end{proof}
		
		It remains to show there are no bad cuts in $H$. Suppose for a contradiction $T \subset E(H)$ is a minimal bad cut. 
		Since $H$ is a planar graph, $T^*$ is a cycle in the dual graph $H^*$.
	
		For a face $z_i$ in $H$, we denote its dual vertex in $H^*$ by $z_i^*$. Recall the dual of an edge $z^*x^* \in E(H^*)$ is a well-defined edge that is contained in both boundaries of $z,x \in F(H)$. 
		Suppose the edges of $T^*$, in order, are $\{z^*_1z^*_2, z^*_2z^*_3, \dots, z^*_kz^*_1\}$. Then $(z_1, \dots, z_k) \subseteq F(H)$ is a sequence of distinct faces of $H$  such that $T = \{e_1, e_2, \dots, e_k\}$, where $e_i$ is a well-defined edge on the boundaries of both $z_i$ and $z_{i+1}$, and $e_k$ is on the boundaries of both $z_k$ and $z_1$.
		
		\begin{figure}[H]
			\centering 
			\includegraphics[scale=0.7]{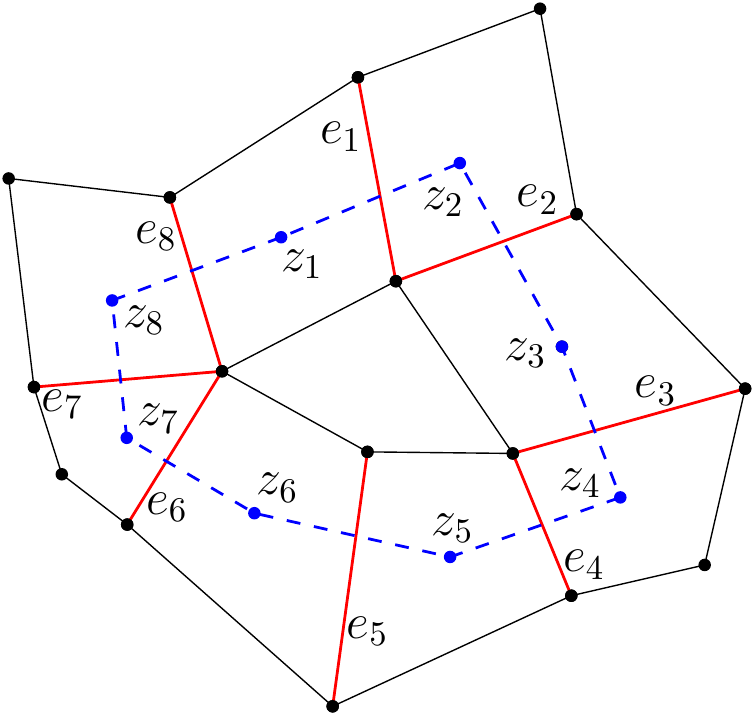}
			\caption{An illustration of what $T$ looks like in $H$ with respect to the dual. A subgraph of $H$ is shown in black, with edges of $T$ highlighted in red. Each $z_i$ denotes a face in $H$, and correspond to a vertex $z_i^*$ in the dual. $T^*$ is shown in dashed blue. (Note the vertices are in the correct relative locations but do not necessarily reflect a proper circle packing representation.)}
			\label{fig:duality2}
		\end{figure}
	
		Consider $H[T]$, the subgraph induced by the edges of $T$: Since $F$-vertices in $H[T]$ have degree one, the components of $H[T]$ must be star graphs. If there is only one component, say with center $u$ and leaves $N(u)$, then $T$ disconnects $u$ from the rest of the graph, contradicting the fact that $u$ must have a good neighbour. Hence there must be at least two components in $H[T]$. (For example, in Figure~\ref{fig:duality2}, $H[T]$ is in red and consists of 4 components.)
		
		Suppose $e_i$ and $e_{i+1}$ are in distinct components of $H[T]$. Both edges are on the boundary of face $z_i$. By Claim~\ref{claim:H-boundary-good}, we know $z_i$ is not the unbounded face of $H$. Recall each bounded face of $H$ has size four, hence we may denote the four vertices on the boundary of $z_i$ by $u,f,v,g$, where $u,v$ are $V$-vertices and $f,g$ are $F$-vertices. Suppose without loss of generality $e_i = uf$. Then since $e_i$ and $e_{i+1}$ are not connected, we must have $e_{i+1} = vg$.
		
		\begin{figure}[H]
			\centering 
			\includegraphics[scale=0.6]{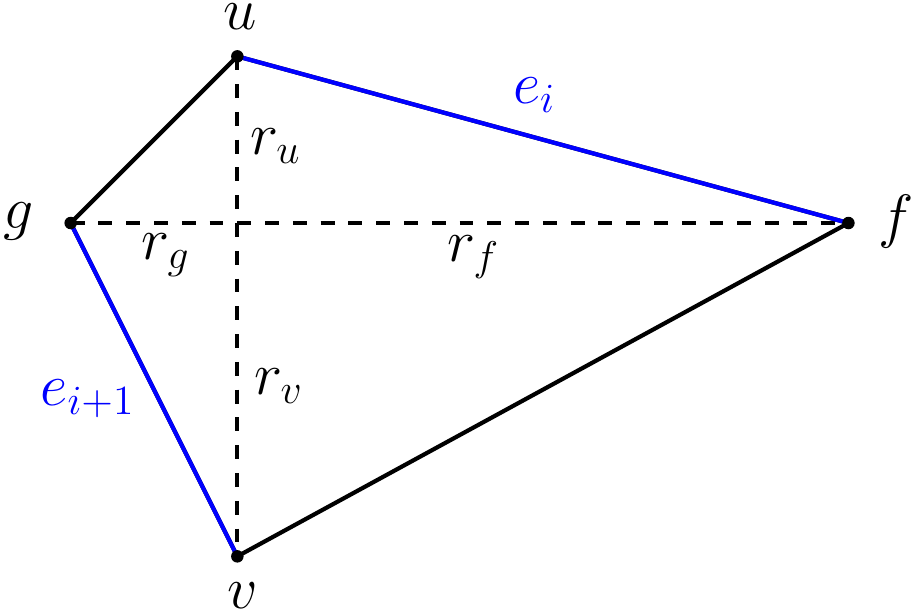}
			\caption{The face $z_i$ in $H$. Note the radii are not necessarily accurate.}
		\end{figure}
		
		Recall $r_u > 2n \cdot r_f$ and $r_v > 2n \cdot r_g$ by definition of bad edges and the fact that $F$-vertices only have big neighbours. 
		Consider the edge $fv$: 
		\begin{enumerate}
			\item If $r_v > 2n \cdot r_f$, then both $u$ and $v$ are large neighbours of $f$;
			\item If $n^2 \cdot r_f \geq r_v$, then $r_u > 2n \cdot r_f \geq r_v > 2n \cdot r_g$, so both $u$ and $v$ are large neighbours of $g$.
		\end{enumerate}
		In both cases, we get a contradiction to Claim~\ref{claim:f-vertex-neighbour}. It follows that there are no bad cuts in $H$, which concludes the overall proof.
	\end{proof}

	As a corollary of Theorem~\ref{thm:good-tree}, we have the following:
	\begin{corollary}[\cite{MoharEuclidean}]\label{cor:min-max-radii-ratio}
		Let $G$ be a triangulation with $|V(G)|+|V(G^*)| = n$. Let $r$ be the radii vector of a valid primal-dual circle packing for $G$. Then $r_{\max} / r_{\min} \leq (2n)^n$. \qed
	\end{corollary}

	We remark here that if the maximum degree of $G$ is $\Delta$, then $2n$ in the definition of good edge can be replaced by $2\Delta$, and the good tree proof would still hold true. Furthermore, for any edge $vf \in E(H)$, we can show there is a good path from $v$ to $f$ of length $O(\Delta)$ by a careful case analysis around vertex $v$ similar to above. It follows that $r_{\max}/r_{\min} \leq \Delta^{O(\Delta D)}$ where $D$ is the diameter of $G$. The proof is omitted.

	Finally, although we assume in this section that the original graph $G$ is a triangulation, we conjecture the analogous result holds for general graphs:
	\begin{conjecture}
		Let $G$ be a 3-connected planar graph, and let $\hat{H}_G$ be its angle graph. Suppose $r$ is the radii vector of a valid primal-dual circle packing representation for $G$. Then there exists a good tree in $\hat{H}_G$ with respect to $r$.
	\end{conjecture}
	
\section{Computing the Primal-Dual Circle Packing}

Throughout this section, $G$ denotes the triangulation
with outer cycle $C_o=(s_1,s_2,s_3)$ and unbounded face $f_{\infty}$ given as input to the algorithm;
$\hat{H}$ denotes the angle graph
of $G$ and $H$ the reduced angle graph. Let $n =|V(G)|+|F(G)|-1 =|V(H)|$. We index vectors by vertices rather than integers.
Our goal is to compute the radii for the $C_o$-regular representation of $G$. Recall $r_{f_\infty}, r_{s_1}, r_{s_2}, r_{s_3}$ are fixed by definition of $C_o$-regularity.

\subsection{Convex Formulation}

We transform the combinatorial question of finding the radii into a minimization problem of a continuous
function. A variant of this formulation was first given in~\cite{de1991principe}. 
\begin{defn}
Consider the following convex function $\Phi$ over $\mathbb{R}^{V(H)\setminus C_o}$:
\begin{equation}
\Phi(x):=\sum_{uw\in E(H)}\left(F(x_{u}-x_{w})+F(x_{w}-x_{u})-\frac{\pi}{2}(x_{u}+x_{w})\right)+2\pi\sum_{u\in V(H)}x_{u}\label{eqn:convex-objective}
\end{equation}
where $F(x)=\int_{-\infty}^{x}\arctan(e^{t})dt$,
and instances of $x_{s_{i}}$ in the expression take constant value of $\log\tan(\frac{\pi}{3})$ for all $s_i \in C_o$.
\end{defn}

The construction of this function $\Phi$ is motivated
by the optimality condition at its minimum $x^{*}$: for all $u\in V(H)\setminus C_o$,
\begin{align}
0 = \frac{\partial\Phi}{\partial x_{u}}(x^{*}) & =\sum_{w\;:\;uw\in E(H)}\left(F'(x_{u}^{*}-x_{w}^{*})-F'(x_{w}^{*}-x_{u}^{*})-\frac{\pi}{2}\right)+2\pi\nonumber \\
 & =\sum_{w\;:\;uw\in E(H)}\left(\arctan(e^{x_{u}^{*}-x_{w}^{*}})-\arctan(e^{x_{w}^{*}-x_{u}^{*}})-\frac{\pi}{2}\right)+2\pi\nonumber \\
 & =-2\sum_{w\;:\;uw\in E(H)}\arctan(e^{x_{w}^{*}-x_{u}^{*}})+2\pi\label{eq:circle_packing}
\end{align}
where we used that $\arctan(u)+\arctan(\frac{1}{u})=\pi/2$
for all $u>0$ at the end. Hence, $\exp(x^{*})$ satisfies the angle constraints
from Equation~\eqref{eqn:angle-constraints} for all $u\in V(H)\setminus C_o$.

\subsection{Correctness}

To show the minimizer of $\Phi$ gives a primal-dual circle-packing,
we first show $\Phi$ is strictly convex, which implies the solution
is unique.
\begin{lemma}
\label{lem:strictly_convex}For any $x$,
\[
\nabla^{2}\Phi(x)=\sum_{uw\in E(H)}2F''(x_{u}-x_{w})b_{uw}b_{uw}^{\top}
\]
where $b_{uw}\in\mathbb{R}^{V(H)\setminus C_o}$ is the vector
with 1 in the $u$ entry, $-1$ in the $w$ entry, and zeros everywhere
else. If $u$ or $v$ or both belongs to $C_o$, then $b_{uv}$ has only one or no non-zero entries. Furthermore, 
\[
\nabla^{2}\Phi(x)\succcurlyeq\frac{2}{n^{2}}\cdot\min_{uw\in E(T)}F''(x_{u}-x_{w})\cdot I\succ0
\]
for any spanning tree $T \subset H$.
\end{lemma}

\begin{proof}
The formula of $\nabla^{2}\Phi(x)$ follows from direct calculation.
To prove $\nabla^{2}\Phi(x)$ is positive-definite, we pick any spanning
tree $T$ in $H$. Note that
\begin{align}
\nabla^{2}\Phi(x) & \succcurlyeq\sum_{uw\in E(T)}2F''(x_{u}-x_{w})b_{uw}b_{uw}^{\top}\nonumber \\
 & \succcurlyeq\left(\min_{uw\in E(T)}2F''(x_{u}-x_{w})\right)\sum_{uw\in E(T)}b_{uw}b_{uw}^{\top}.\label{eq:Hess_lower_pf}
\end{align}
Fix any $h\in\mathbb{R}^{V(H)\setminus C_o}$ with $\|h\|_{2}^{2}=1$.
Then
\[
\sum_{uw\in E(T)}(b_{uw}^{\top}h)^{2}=\sum_{uw\in E(T)}(h_{u}-h_{w})^{2},
\]
where we define $h_{s}=0$ for all $s \in C_o$. Since $\|h\|_{2}^{2}=1$,
there exists a vertex $v$ such that $h_{v}\geq\frac{1}{\sqrt{n}}$. Now, consider
the path $P$ from $v$ to some $s \in C_o$. We have
\[
\sum_{uw\in E(T)}(b_{uw}^{\top}h)^{2}\geq\sum_{uw\in P}(h_{u}-h_{w})^{2}\geq\sum_{uw}(\frac{1}{\sqrt{n}|P|})^{2}=\frac{1}{n|P|}\geq\frac{1}{n^{2}},
\]
where we used the fact that the minimum of $\sum_{uw\in P}(h_{u}-h_{w})^{2}$
is attained by the vector $h$ whose entries decrease from $h_{v}=\frac{1}{\sqrt{n}}$ to $h_{s}=0$
uniformly on the path $P$. Using this in \eqref{eq:Hess_lower_pf}, we
have that for any $h$ with $\|h\|_2^{2}=1$,
\begin{align*}
h^{\top}\nabla^{2}\Phi(x)h &\geq \frac{2}{n^{2}}\cdot\min_{uw\in E(T)}F''(x_{u}-x_{w}) \\
\intertext{Since $F''(x)=\frac{\exp(x)}{\exp(2x)+1}>0$
	for all $x$, we have}
\nabla^{2}\Phi(x)&\succcurlyeq\frac{2}{n^{2}}\cdot\min_{uw\in E(T)}F''(x_{u}-x_{w})\cdot I\succ0.
\end{align*}
This proves that $\Phi$ is strictly convex.
\end{proof}
Now, we prove that the minimizer of $\Phi$ is indeed a primal-dual
circle packing.
\begin{theorem}
\label{thm:correctness}Let $x^{*}$ be the minimizer of $\Phi$.
Then, $r^{*}=\exp(x^{*})$, where the exponentiation is applied coordinate-wise,
is the radii vector of the unique $C_o$-regular primal-dual circle packing representation of $G$.
\end{theorem}

\begin{proof}
As discussed in Section~$\ref{subsec:angle-graph}$, there exists a unique $C_o$-regular circle packing representation. Theorem~\ref{thm:r-existence} shows that the associated radii vector $r$ satisfies
\[
\sum_{w\;:\;uw\in E(H)}\arctan(r_{w}/r_{u})=\pi,
\]
for all $u\in V(H)\setminus C_o$. By the formula of $\nabla\Phi$
in Equation~\eqref{eq:circle_packing}, we know $\nabla\Phi(\log r)=0$; therefore, $r$ is a minimizer of $\Phi$. Since $\Phi$
is strictly convex by Lemma~\ref{lem:strictly_convex}, the minimizer
is unique. Hence, $r^* = r$. 
\end{proof}

\subsection{Algorithm for Second-Order Robust Functions}

To solve for the minimizer of $\Phi$, a convex programming
result is used as a black box. We define the relevant terminology
below, and then present the theorem.
\begin{defn}
A function $f$ is \emph{second-order robust} with respect to $\ell_{\infty}$
if for any $x,y$ with $\norm{x-y}_{\infty}\leq1$,
\[
\frac{1}{c}\grad^{2}f(x)\preccurlyeq\grad^{2}f(y)\preccurlyeq c\grad^{2}f(x)
\]
for some universal constant $c>0$.
\end{defn}

Intuitively, the Hessian of a second-order robust function does not
change too much within a unit ball.
\begin{theorem}[{\cite[Thm 3.2]{cohen2017matrix}}]
\label{thm:convex-bb}Let $g:\mathbb{R}^{n}\to\mathbb{R}$ be a second-order
robust function with respect to $\ell_{\infty}$, such that its Hessian
is symmetric diagonally dominant (SDD) with non-positive off-diagonals,
and has $m$ non-zero entries. Given a starting point $x^{(0)}\in\mathbb{R}^{n}$,
we can compute a point $x$ such that $g(x)-g(x^{*})\leq\eps$ in
expected time
\[
\widetilde{O}\left((m+T)(1+D_{\infty})\log\left(\frac{g(x^{(0)})-g(x^{*})}{\eps}\right)\right)
\]
where $x^{*}$ is a minimizer of $g$, $D_{\infty}=\sup_{x:g(x)\leq g(x^{(0)})}\norm{x-x^{(0)}}_{\infty}$
is the $\ell_{\infty}$-diameter of the corresponding level-set of
$g$, and $T$ is the time required to compute the gradient and Hessian of $g$.
\end{theorem}

The algorithm behind the above result essentially uses Newton's method
iteratively, each time optimizing within a unit $\ell_{\infty}$-ball.
The key component involves approximately minimizing a
SDD matrix with non-positive off-diagonals in nearly linear time, by recursively approximating Schur complements.

For our function $\Phi$, there are two difficulties in using this
theorem. First, the level-set diameter $D_\infty$ could be very large because $\Phi$ is
only slightly strongly-convex.
So it would be better if $D_\infty$ were replaced with the distance between $x^{(0)}$ and $x^{*}$.
Second, we are multiplying $D_{\infty}$ and $\log(1/\eps)$ in the run-time expression when both terms could be very large; we would like to add the two instead. It turns out both can
be achieved at the same time by modifying the objective.
\begin{theorem}
\label{thm:improved-convex-bb}Let $g:\mathbb{R}^{n}\to\mathbb{R}$
be a second-order robust function with respect to $\ell_{\infty}$,
such that its Hessian is symmetric diagonally dominant with non-positive
off-diagonals, and has $m$ non-zero entries. Let $x^*$ be the minimizer of $g$, and suppose that $\nabla^{2}g(x^{*})\succcurlyeq\alpha I$, for some $\alpha < 1$.
Given a starting point $x^{(0)}\in\mathbb{R}^{n}$ and any $\eps\leq\alpha/2$, we can compute a point $x$ such that $g(x)-g(x^{*})\leq\eps$
in expected time
\[
\widetilde{O}\left((m+T)\left(R_{\infty}\log^2\left(\frac{g(x^{(0)})-g(x^{*})}{\alpha}\right)+\log\left(\frac{\alpha}{\eps}\right)\right)\right)
\]
where $R_{\infty}=\norm{x^{*}-x^{(0)}}_{\infty}$, and $T$ is the time required
to compute the gradient and Hessian of $g$. Furthermore, we have that $\|x-x^{*}\|_{2}^{2}\leq \eps/\alpha.$
\end{theorem}

\begin{proof}
    The algorithm builds on Theorem \ref{thm:convex-bb}, and the high
  level idea can be broken into two steps: The first step
  transforms the dependence on the diameter of the level set in
  Theorem \ref{thm:convex-bb} to the $\ell_{\infty}$ distance
  $\norm{x^{(0)} - x^*}_{\infty}$ from the initial point; the second
  step leverages the strong-convexity at the minimum to obtain an
  improved running time.

  For the first step, given the function $g(x)$ and an initial point
  $x^{(0)}$, we construct an auxiliary function $\tilde{g}(x)$ that
  adds a small convex penalty reflecting the distance between $x$ and
  the initial point $x^{(0)}$. Analytically, this allows us to replace
  the dependency on the diameter of the level-set in Theorem
  \ref{thm:convex-bb} with the initial $\ell_{\infty}$-distance
  $\norm{x^{(0)} - x^*}$.

  The second step leverages the fact that $\grad g(x^*) \succcurlyeq \alpha I$ at the minimum. Since the Hessian of $g$ is also robust, it
  is $\succcurlyeq \Omega(\alpha) I$ near the minimum.  Strong convexity implies that the additive error at a point is
  proportional to the distance from the point to $x^*$. Hence,
  running a robust Newton's method to roughly $\alpha$ additive accuracy
  guarantees that the output point $x^{(1)}$ is within an
  $\ell_{\infty}$ distance of roughly 1 from the minimum. The run-time to this point is less than running to $\eps$-accuracy when
  $\alpha > \eps$.  
  We then run the algorithm a second time to
  $\eps$-accuracy starting from $x^{(1)}$; this instance has a much  reduced $R_{\infty}$ distance. The run-time for the two phases together is lower compared to running the algorithm just once starting from $x^{(0)}$.

  The above overview is informal; in particular, the two steps
  cannot be as cleanly separated as described. Indeed, when constructing the auxiliary function $\tilde{g}$, we
  require prior knowledge of the initial distance
  $R_{\infty} = \norm{x^{(0)} - x^*}_{\infty}$ within a constant factor. To overcome this, we use a standard doubling trick: Starting
  from a safe lower bound, we presuppose an estimate for $R_{\infty}$
  and run the two steps as above. If the estimate for $R_{\infty}$ was
  too small (which we can detect), then we double our guess and try
  again. Since the overall run-time is proportional to $R_{\infty}$, we do not add to it asymptotically.

  We now describe the algorithm and prove the theorem in full detail.
  To begin, suppose $R_{\infty}:=\norm{x^{*}-x^{(0)}}_{\infty}$ is
  given. To minimize $g$, we construct a new function
\[
\tilde{g}(x)=g(x)+\frac{\eps}{4n}\sum_{i}\cosh\left(\frac{x_{i}-x_{i}^{(0)}}{R_{\infty}}\right).
\]
Note that if $x^\dagger$ is the minimizer of $\tilde g$, then
\[
\tilde g(x^\dagger) = \min_{x}\tilde{g}(x)\leq g(x^{*})+\frac{\eps}{4n}\sum_{i}\cosh\left(\frac{x_{i}^*-x_{i}^{(0)}}{R_{\infty}}\right)\leq g(x^{*})+\frac{\eps}{4n}\sum_{i}\cosh(1)\leq g(x^{*})+\frac{\eps}{2}
\]
and $\tilde{g}(x)\geq g(x)$ for all $x$. Therefore, to minimize $g$ with
$\eps$ accuracy, it suffices to minimize $\tilde{g}$ with $\eps/2$ accuracy.

We check the condition of Theorem \ref{thm:convex-bb} for $\tilde g$. The Hessian of $\tilde{g}$ is simply the Hessian of $g$ plus
a diagonal matrix, so $\nabla^{2}\tilde{g}$ is still SDD
with non-positive off-diagonals.
A simple calculation shows that $\tilde{g}$ is second-order robust.
To bound $D_{\infty}:=\sup_{x:\tilde{g}(x)\leq\tilde{g}(x^{(0)})}\norm{x-x^{(0)}}_{\infty}$,
note that for any $x$ with $\tilde{g}(x)\leq\tilde{g}(x^{(0)})$, we
have
\[
g(x^{(0)}) + \eps/4 = \tilde g(x^0) \geq g(x)+\frac{\eps}{4n}\sum_{i}\cosh\left(\frac{x_{i}-x_{i}^{(0)}}{R_{\infty}}\right)\geq g(x^{*})+\frac{\eps}{8n}\exp\left(\frac{\|x-x^{(0)}\|_{\infty}}{R_{\infty}}\right).
\]
Hence,
\[
D_\infty = \sup_{x:\tilde{g}(x)\leq\tilde{g}(x^{(0)})}\|x-x^{(0)}\|_{\infty}\leq R_{\infty}\log\left(\frac{8n}{\eps}( g(x^{(0)})-g(x^{*}) + \eps/4)\right).
\]

We apply Theorem~\ref{thm:convex-bb} to $\tilde g$ to get a point $x$ such that $\tilde g(x) - \tilde g(x^\dagger) < \eps/2$, using time
\begin{align*}
& \widetilde{O}\left((m+T)(1+D_{\infty})\log\left(\frac{\tilde g(x^{(0)})-\tilde g(x^{\dagger})}{\eps/2}\right)\right) \\
= & \widetilde{O}\left((m+T)\left(1+\log\left(\frac{g(x^{(0)})-g(x^{*})}{\eps}\right)R_{\infty}\right)\log \left(\frac{g(x^{(0)})-g(x^{*})}{\eps}\right)\right).
\end{align*}
This $x$ minimizes $g$ to $\eps$ accuracy. Henceforth we view the above reduction from $g$ to $\tilde{g}$ as a black-box. Now, we make some further observations regarding $g$.

\begin{lemma}\label{lem:convexity-bound}
	For any constant $C \leq 1$ and $x$ such that $\norm{x - x^*}_\infty = C$, we have $g(x) \geq g(x^*) + \Omega(\alpha) \cdot C^2$. Furthermore, if $x'$ satisfies $g(x') - g(x^*) \leq o(\alpha) \cdot C^2$, then $\norm{x' - x^*}_\infty \leq C$.
\end{lemma}
\begin{proof}
	Since $\nabla^{2}g(x^{*})\succcurlyeq\alpha\cdot I$ and $g$
	is second-order robust,  $\nabla^{2}g(x)\succcurlyeq\Omega(\alpha)\cdot I$
	for all $x$ with $\norm{x-x^{*}}_{\infty}\leq1$. Applying the Mean Value Theorem for $x$ with $\norm{x - x^*}_\infty = C$, we get
	\begin{equation} \label{eq:lowerbound-away-from-min}
	g(x)\geq g(x^{*})+\Omega(\alpha)\cdot\|x-x^{*}\|_{2}^{2}\geq g(x^{*})+\Omega(\alpha) \cdot C^2.
	\end{equation}
	Moreover, by convexity of $g$, we have $g(x)\geq g(x^{*})+\Omega(\alpha) \cdot C^2$
	for all $x$ where $\|x-x^{*}\|_{\infty}\geq C$. The second part of the Lemma is the contrapositive.
\end{proof}

To achieve the run-time stated in the theorem, we minimize $g$ in two phases.
In the first phase, we use $\eps_1 = \alpha/\log^2(\alpha/\eps)$ and initial point $x^{(0)}$ as given, to get a point $x^{(1)}$ such that $g(x^{(1)}) - g(x^*) \leq \eps_1$.
By Lemma~\ref{lem:convexity-bound}, we have $\|x^{(1)}-x^{*}\|_{\infty}\leq 1/\log(\alpha/\eps)$.
In the second phase, we minimize to $\eps$ error. However, since  $x^{(1)}$ can be used as the initial point, we know $R_\infty = 1/\log(\alpha/\eps)$. The algorithm returns $x^{(2)}$ such that $g(x^{(2)}) - g(x^*) \leq \eps$. Summing the run-time of the two phases carefully, we get the desired total time,
where factors of $\log\log(\alpha/\eps)$ are hidden. The claim $\norm{x^{(2)} - x^*}_2^2 \leq \eps/\alpha$ follows from  Equation~\eqref{eq:lowerbound-away-from-min} in  Lemma~\ref{lem:convexity-bound}.

Finally, we resolve the initial assumption of $R_{\infty}$ being given.
Note that we only use $R_{\infty}$ during the first phase of the algorithm, where we use the target accuracy
$\eps_1$. To run the first phase without knowing $R_{\infty}$,
we apply Lemma~\ref{lem:convexity-bound} again in a doubling trick.

Let $x^{(r)}=\arg\min_{\|x\|_{\infty}\leq r}g(x)$.
Consider $\hat{x} = \frac{x^{(r)}-x^{*}}{\|x^{(r)}-x^{*}\|_{\infty}}$, which satisfies $\norm{x^* - (x^* + \hat{x})}_\infty = 1$. Hence, by Lemma~\ref{lem:convexity-bound},
\[
g(x^* + \hat{x}) \geq g(x^{*})+\Omega(\alpha).
\]
Furthermore, $x^* + \hat{x}$ is on the straight line connecting $x^{*}$ and $x^{(r)}$. Since the slope of $g$ is increasing from $x^{*}$ to $x^{(r)}$, if
$\|x^{*}\|_{\infty}>2r$, we also have
\[
g(x^{(r)})\geq g(x^{(r)} - \hat{x})+\Omega(\alpha).
\]
Note that $\norm{x^{(r)} - \hat{x}}_\infty \leq 2r$. This shows that $\|x^{*}\|_{\infty}>2r$ implies
\[
\min_{\|x\|_{\infty} \leq r}g(x)\geq \min_{\|x\|_{\infty}\leq2r}g(x)+\Omega(\alpha).
\]
Hence, if $R_{\infty}>2r$, then we can detect it by comparing $x^{(r)}$ and $x^{(2r)}$.
To estimate $R_\infty$, first we run the algorithm while pretending $R_{\infty}=1$
and compare the result against $R_{\infty}=2$. If it fails this test,
then we compare the result for $R_{\infty}=2$ against $R_{\infty}=4$, and so on. We stop when the test passes, at which point the guess for $R_\infty$ is correct to a constant factor, and the true $x^*$ has been found.
This does not
affect the run-time asymptotically,
since the total time simply involves a term $1+2+\cdots+R_{\infty}$ instead of $R_{\infty}$.
\end{proof}

\subsection{Strong Convexity at the Minimum}

To apply Theorem \ref{thm:improved-convex-bb}, we need to show $\Phi$ is strongly-convex at $x^{*}$.
\begin{lemma}
\label{lem:strongly_convex}Let $x^{*}$ be the minimizer of $\Phi$.
Then,
\[
\nabla^{2}\Phi(x^{*})\succcurlyeq\frac{1}{n^{3}}I.
\]
\end{lemma}

\begin{proof}
Lemma \ref{lem:strictly_convex} shows that
\begin{equation}
\nabla^{2}\Phi(x)\succcurlyeq\frac{2}{n^{2}}\cdot\min_{uw\in E(T)}F''(x_{u}-x_{w})\cdot I\label{eq:Hess_low}
\end{equation}
for any spanning tree $T\subset H$. Theorem \ref{thm:good-tree}
shows that there is a spanning tree $T$ such that $|x_{u}-x_{w}|\leq \log 2n$
for any $uw\in E(T)$. Hence, we have
\[
F''(x_{u}-x_{w})\geq\exp(-|x_{u}-x_{w}|)\geq (2n)^{-1}
\]
for any $uw\in E(T)$. Putting it into Equation~\eqref{eq:Hess_low} gives
the desired bound.
\end{proof}

\subsection{Result}
We combine the previous sections for the overall result.
\begin{theorem}
	\label{thm:cvx-runtime} Let $r \in \R^{V(H) \setminus C_o}$ be the radii of the $C_o$-regular primal-dual circle packing representation of the triangulation $G$. Let $x^* = \log r$. For any $0<\eps<\frac{1}{2}$, we can compute a point
	$x$ such that $\|x-x^{*}\|_{\infty}\leq\eps$ in expected time 
	\[
	\widetilde{O}\left(n\log\frac{R}{\eps}\right),
	\]
	where $R = r_{\max}/r_{\min}$ is the ratio of the maximum to minimum radius of the circles.
\end{theorem}

\begin{proof}
	We check the conditions of Theorem \ref{thm:improved-convex-bb}.
	Lemma \ref{lem:strictly_convex} shows that
	\[
	\nabla^{2}\Phi(x)=\sum_{uw\in E(H)}2F''(x_{u}-x_{w})b_{uw}b_{uw}^{\top}.
	\]
	Since $F''>0$, we know $\nabla^{2}\Phi(x)$ is a positive combination of
	$b_{uw}b_{uw}^{\top}$ (each SDD with non-positive off-diagonals). Hence,
	$\nabla^{2}\Phi(x)$ is an SDD matrix with non-positive off-diagonals. 
	
	To show second-order robustness, note that if $x$ changes by 
	at most $1$ in the $\ell_{\infty}$-norm, then $x_{u}-x_{w}$ changes by at most $2$. Recall
	\[
	F''(x_{u}-x_{w})=\frac{\exp(x_{u}-x_{w})}{\exp(2(x_{u}-x_{w}))+1},
	\]
	so $F''$ changes by at most a factor of $e^{2}$. 
	It follows that $\nabla^{2}\Phi$ changes by at most a constant multiplicative factor.
	
	Lemma~\ref{lem:strongly_convex} shows that $\nabla^{2}\Phi(x^{*})\succcurlyeq\alpha I$
	with $\alpha=n^{-3}$. 
	
	Now, we can apply Theorem~\ref{thm:improved-convex-bb} as a black-box.
	Since the Hessian has $O(n)$ entries, each of which consists of a constant number of 
	hyperbolic computations, both $m$ and $T$ are $O(n)$.
	A simple initial point $x^{(0)}$ is the all zeros vector; it follows that
	$\Phi(x^{(0)})=O(n)$.
	
	$R_{\infty}$ in Theorem~\ref{thm:improved-convex-bb} is precisely $\norm{x^*}_\infty$.
	Since $x_{u}^{*} = \log r^*_u \leq 0$ for all $u\in V(H) \setminus C_o$,
	we know $x^{*}$ satisfies $\norm{x^{*}}_{\infty} = -\log r^*_{\min}<\log(r^*_{\max}/r^*_{\min})$,
	where $r^*_{\min}$ is the radius of the smallest circle in the true circle packing representation, and $r^*_{\max}=\tan(\frac{\pi}{3})$ is the radius of the largest circle, attained by vertices in $C_o$. 
	Therefore $R_{\infty}=\norm{x^{*}}_{\infty}\leq\log R$. 
	
	Lastly we estimate $\Phi(x^{*})$. Note that $\frac{\pi}{2}|z|\leq F(z)+F(-z)\leq\frac{\pi}{2}|z|+2$
	for all $z\in\mathbb{R}$. By Corollary~\ref{cor:min-max-radii-ratio}, $\norm{x^*}_\infty = \widetilde O(n)$ in the worst case. Hence,
	\begin{align*}
		-\Phi(x^{*}) & \leq - \sum_{uw\in E(H)}\left(\frac{\pi}{2}|x_{u}^{*}-x_{w}^{*}|-\frac{\pi}{2}(x_{u}^{*}+x_{w}^{*})\right)-2\pi\sum_{u\in V(H)}x_{u}^{*}\leq \widetilde{O}(n^{2}).
	\end{align*}
	It follows that $\Phi(x^{(0)})-\Phi(x^{*})\leq\widetilde{O}(n^{2})$.
	
	Now, Theorem~\ref{thm:improved-convex-bb} shows how to find $x$ with
	$\|x-x^{*}\|_{\infty}\leq\eps$ in time
	\[
	\widetilde{O}\left(n\left(\log R\log n+\log\frac{n}{\eps}\right)\right)=\widetilde{O}\left(n\log\frac{R}{\eps}\right).
	\]
\end{proof}
\subsection{Computing the Locations of the Vertices}\label{subsec:compute-loc}
After approximating the radii, we embed the primal and dual vertices using the reduced angle graph $H$. We emphasize at this point that $H$ is already a plane graph, so the cyclic ordering of neighbours around each vertex is known.

Suppose $r$ is the $\eps$-approximation of the radii we obtained, and $r^*$ is the true radii vector of the $C_o$-regular primal-dual representation. We define an edge $uw$ in $H$ to be \emph{approximately-good} (with respect to $r$) if it satisfies $ (1-\eps)/(1+\eps) \cdot (2n)^{-1} \leq r_u/r_w \leq (1+\eps)/(1-\eps) \cdot 2n$, and \emph{approximately-bad} if it does not. Note that a good edge (with respect to $r^*$) is an approximately-good edge, and an approximately-bad edge is a bad edge.

Recall the outer cycle of $G$ is $C_o = (s_1, s_2, s_3)$, and let the outer cycle of $H$ be denoted by $B = (s_1, t_1, s_2, t_2, s_3, t_3)$. We may assume for both the true embedding and our appropximate embedding, that $s_1$ is positioned at the origin, and $s_1s_2$ lie on the $x$-axis.

The high-level idea is to embed the vertices one-by-one following a breadth-first style traversal through $H$ using only approximately-good edges. Since the true positions of $s_1, s_2, s_3 \in C_o$ are known, and the outer cycle of $H$ consists of good edges, $t_1, t_2, t_3$ are embedded first. We proceed in a breadth-first fashion with one additional traversal restriction: Suppose we visited the vertex $u$, and let the neighbours of $u$ be $w_1, \dots, w_m$ in cyclic order. We can visit a neighbour $w_i$ only if either $w_{i-1}$ or $w_{i+1}$ has been visited already. This is so that when we move from $u$ to an unvisited neighbour $w_i$ (suppose $w_{i-1}$ was visited), we can place the kite $K_{uw_i}$ (see Section~\ref{subsec:angle-graph}) with one point at $u$ and one side tangent to the previous kite $K_{uw_{i-1}}$. The kite in turn determines the position of $w_i$ in the approximate circle packing representation. First, we show all vertices in $H$ can be reached this way.

\begin{figure}[H]
	\centering 
	\includegraphics[scale=0.6]{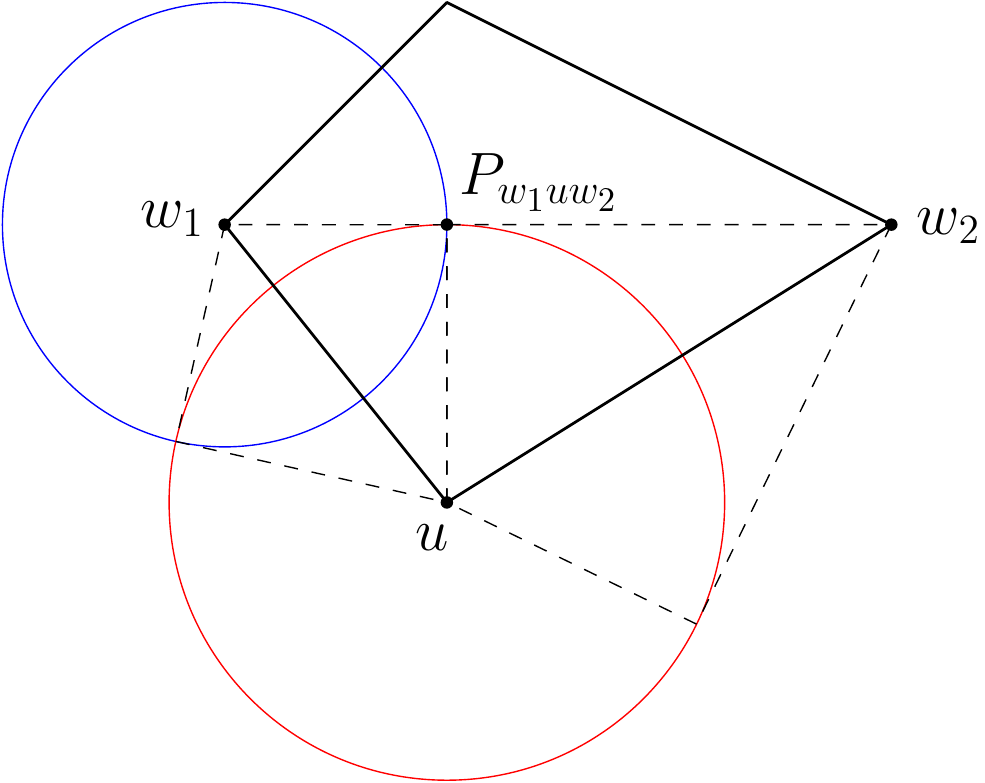}
	\caption{Embedding $H$ locally around $u$. The two kites $K_{uw_1}$ and $K_{uw_2}$ are shown in dashed lines. After $u$ and $w_1$ are embedded, the kite $K_{uw_2}$ is placed with one vertex at $u$, and one side tangent to $K_{uw_1}$ along the line $uP_{w_1uw_2}$. Its side lengths are $r_u$ and $r_{w_2}$.} 
	\label{fig:location}
\end{figure}

Suppose the vertex $w$ has neighbours $v_1, \dots, v_m$ in cyclic order, and we visited $v_i$ from $w$ but cannot reach $v_{i+1}$, due to the fact that $wv_{i+1}$ is an approximately-bad edge. Observe that $w, v_i, v_{i+1}$ are in a face together with another vertex $x_2$ (recall any bad edge is on the boundaries of two faces of degree four). By the arguments in Section~\ref{subsec:good-tree} and a simple case analysis, we see that either we can reach $v_{i+1}$ from $w$ by going through $v_i$ and then $x_2$ (implying $wv_i, v_ix_2, xv_{i+1}$ are all approximately-good edges), or $v_{i+1}$ is an $V$-vertex, $v_{i+1}x_2, v_{i+1}w$ are bad edges, and $v_i x_2$ is a good edge. In the latter case, let the neighbours of $v_{i+1}$ be $w = x_1, x_2, \dots, x_l$ in cyclic order, and suppose $j \geq 3$ is the smallest index at which $v_{i+1}x_j$ is a good edge. Note that for each $k < j$, the vertices $v_{i+1}, x_k, x_{k+1}$ are in a face together with another vertex $y_k$. As $v_{i+1}$ is a $V$-vertex, all the $x$'s are $F$-vertices of degree three; furthermore, since $v_{i+1}x_k$ is a bad edge, we know $x_ky_k$ and $x_k y_{k+1}$ must be good and hence approximately-good. (See Figure~\ref{fig:traversal} for an example with $j=5$.) It follows that $(w = x_1, y_1, x_2, y_2, \dots, y_{j-1}, x_j, v_{i+1})$ is an approximately-good path from $w$ to $v_{i+1}$, and going along this path does not violate our traversal restrictions. So we have shown that all vertices in $H$ can be reached.

\begin{figure}[H]
	\centering 
	\includegraphics[scale=0.95]{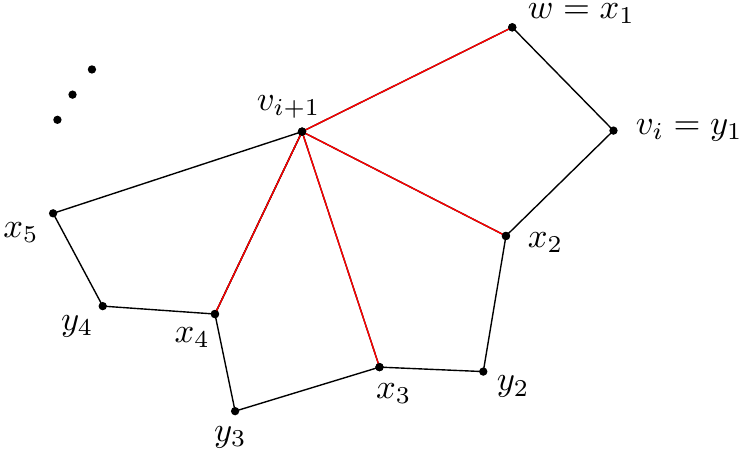}
	\caption{$H$ locally around a vertex $v_{i+1}$. The bad edges are shown in red; they may not be all approximately-bad. However, an approximately-good path from $w$ to $v_{i+1}$ exists and can be followed given the traversal restrictions. (Note that the angles in this diagram do not reflect the correct circle packing representation.)}
	\label{fig:traversal}
\end{figure} 
For any edge $uw$, we can define the angle $\theta_{uw}$ as half the angle contributed by the kite $K_{uw}$ at $u$ (see Section~\ref{subsec:angle-graph}). Compared to the true angle $\theta_{uw}^*$, the error is
\begin{align*}
	\abs{\theta_{uw} - \theta^*_{uw}} &= \abs{\arctan \frac{r_w}{r_u} - \arctan \frac{r_w^*}{r_u^*}} \leq 4\eps \frac{r_w^*}{r_u^*},
\end{align*}
where we used a coarse first order approximation, and the fact that $r_u$ is an $\eps$-approximation of $r^*_u$ for all vertices $u$. Along a good edge, this error is bounded by $8 \eps n$.

For an edge $uw$ whose embedding has been approximated, we can further define $\alpha_{uw}$ as the angle formed by the ray $uw$ (starting from the vertex $u$) and the $x$-axis, and $\alpha_{uw}^*$ as the true angle. Observe that if $\alpha_{uw} - \alpha_{uw}^* = \delta$, then $\alpha_{wu} - \alpha_{wu}^* = \delta$. Furthermore, suppose $uw_1, uw_2$ are approximately-good edges, and that we embed $w_2$ right after $u$ and $w_1$. Then 
\[
|\alpha_{uw_2} - \alpha^*_{uw_2}| \leq |\alpha_{uw_1} - \alpha^*_{uw_1}| + |\theta_{uw_1} - \theta_{uw_1}^*| + |\theta_{uw_2} - \theta_{uw_2}^*|;
\]
in other words, the angle errors accumulate linearly as we traverse through $H$. Since only approximately-good edges are used, we conclude that for all edge $uw$ used in the traversal, $|\alpha_{uw} - \alpha_{uw}^*| \leq O(\eps n^2)$. 

Finally, we compare the approximate position $p$ of the vertices with the true positions $p^*$. Suppose $u$ is embedded in its true position, and $v,w$ are consecutive neighbours of $u$ such that $uv, uw$ are approximately-good edges, and $v$ has been embedded. In embedding $w$, error is introduced by $\alpha_{uw}$ as well as by $r_u$ and $r_w$. Specifically, $p^*_w$ is at a distance of $\sqrt{{r_u^*}^2 + {r_w^*}^2}$ away from $p^*_u$ in the direction given by $\alpha^*_{uw}$, while the approximate position $p_{w}$ will be at a distance of $\sqrt{r_u^2 + r_{w}^2}$ away in the direction $\alpha_{uw}$. Basic geometry shows
\[
\norm{p_w^* - p_w}_2 \leq \lvert\sqrt{{r_u^*}^2 + {r_w^*}^2} - \sqrt{r_u^2 + r_w^2}\rvert + \lvert\alpha_{uw} - \alpha_{uw}^*\rvert \sqrt{r_u^2 + r_{w}^2} \leq (O(\eps) + \lvert\alpha_{uw} - \alpha_{uw}^*\rvert)\sqrt{r_u^2 + r_{w}^2} \leq O(\eps n^2),
\]
where we use the fact that $r_u \leq 1$ for all $u \in V(H) \setminus C_o$. The error accumulates linearly as we embed each vertex, hence $\norm{p_u^* - p_u}_2 \leq O(\eps n^3)$ for all vertices $u$. It follows that if $r$ is an $\eps/(n^3R)$-approximation of the true radii, then we can recover positions $p$ such that $\norm{p_u - p_u^*}_2 \leq \eps/R$ for each vertex $u$. Changing $\eps$ by a polynomial factor of $R$ and $n$ does not affect the run-time in the $O$-notation; this completes the proof of Theorem~\ref{thm:main}.

\subsection{A Remark About Numerical Precision}
In the proof of Theorem \ref{thm:cvx-runtime}, we only apply the algorithm from~\cite{cohen2017matrix} on convex functions $\tilde{g}$ that are well-conditioned; specifically, $n^{-O(1)}\cdot I\preccurlyeq\nabla^{2}\tilde{g}(x)\preccurlyeq n^{O(1)}\cdot I$ for all $x$ the algorithm queries.
In this case, it suffices to perform all calculations in finite-precision with $O(\log(\frac{n R}{\eps}))$ bits (See Section 4.3 in~\cite{cohen2017matrix} for the discussion). Therefore, with $O(\log(\frac{n R}{\eps}))$ bits calculations, we can compute the radius with $(1\pm\eps)$ multiplicative error and the location with $\eps/R$ additive error.

	\section{Computing the Primal Circle Packing} \label{subsec:circle-packing-proof}
Finally, we prove Theorem~\ref{cor:circle-packing} as a corollary of Theorem~\ref{thm:main}.

We may assume $G$ is 2-edge-connected, otherwise, we can simply split the graph at the cut-edge, compute the circle packing representations for the two components separately, and combine them at the edge after rescaling appropriately. 

First, a planar embedding of $G$ can be found in linear time such that all face boundaries are cycles. Next, $G$ can be triangulated by adding a vertex in each face of degree greater than three and connecting it to all the vertices on the face boundary. Let $V^+$ denote the set of additional vertices, and let $T$ denote the resulting triangulation with outer cycle $C_o$. Note that $|V(T)| = O(n)$, since $G$ is planar. 

We then run the primal-dual circle packing algorithm on $T$, which returns radius $r_u$ and position $p_u$ for each $u \in V(T)$, corresponding to an approximate $C_o$-regular representation (See Section~\ref{subsec:angle-graph}). By discarding the dual graph and additional vertices $V^+$, we obtain an $\eps$-approximation of the primal circle packing of $G$.  

The total run-time is $\widetilde{O}(|V(T)| \log R_{PD}/\eps)$, where $R_{PD} = r^*_{PD, \max}/r^*_{PD, \min}$ is the ratio of the maximum to minimum  radius in the target primal-dual circle packing of $T$. Let $R = r^*_{\max}/r^*_{\min}$ be the ratio of the maximum to minimum radius in the target primal circle packing. If $R_{PD} \leq \poly(n) R$, then $\widetilde{O}(|V(T)| \log R_{PD}/\eps) \leq \widetilde{O}(n \log R/\eps)$, and we have the claimed run-time. 

Note that $r^*_{PD,\max} = \tan (\pi/3)$ is attained by one of the vertices on the outer cycle $C_o$ of $T$. By construction of $T$, at least two of the vertices on $C_o$ belong to $V$, hence $r^*_{\max} = r^*_{PD, \max}$. It remains to show that $r^*_{\min}$ is polynomially close to $r^*_{PD, \min}$. 

We will make use of the terminology and lemmas introduced in Section~\ref{subsec:good-tree}. There are two cases to consider:
\begin{enumerate}
	\item $r^*_{PD, \min}$ is attained by some primal vertex $v \in V(T)$. If $v \in V$, then we are done. So we may assume $v \in V^+$. In the reduced graph $H_T$, since $v$ is an $V$-vertex, it at least one good neighbour $f$ by Claim~\ref{claim:no-isolated-v-vertex}; moreover, $f$ has at most one bad neighbour by Claim~\ref{claim:f-vertex-neighbour}, so it has another good neighbour $w \neq v$. Note that $w$ is at distance two from $v$ in $H_T$, implying that $w$ is a neighbour of $v$ in T, and therefore $w \in V$. Using properties of good edges, we get $r^*_w \leq (2n)^2 r^*_v$.
	\item $r^*_{PD, \min}$ is attained by some dual vertex $f \in V(T^*)$. We know $f$ is an $F$-vertex in $H_T$ and has at least two good neighbours by Claim~\ref{claim:f-vertex-neighbour}; moreover, it has at least two neighbours in $V$, since every face in $T$ contains at least two vertices from $V$ on its boundary. It follows that $f$ has a good neighbour $w$ where $w \in V$, so $r^*_w \leq 2n r^*_f$.
\end{enumerate}
In both cases, there exists $w \in V$ such that $r^*_w \leq (2n)^2 r^*_{PD, \min}$. Hence, $r^*_{\min} \leq (2n)^2 r^*_{PD, \min}$, as desired.
	
	\section{Acknowledgment}
	We thank Manfred Scheucher for providing an implementation of an algorithm for generating circle packing figures. We thank Adam Brown for helpful discussions. We thank the anonymous referees for their edits and suggestions.
	
	\bibliographystyle{alpha}
	\bibliography{ref}

\end{document}